\documentclass[11pt]{article} 
\usepackage{times}
\usepackage{hyperref}
\hypersetup{
    colorlinks=false,
    pdfborder={0 0 0},
}
\usepackage{url}
\usepackage{amsmath}
\usepackage{bbm}
\usepackage{amsfonts}
\usepackage{mathtools}
\usepackage{empheq}
\usepackage{amsthm}
\usepackage{mathpazo}
\usepackage{bm}
\usepackage{hhline}
\usepackage{fixltx2e}
\usepackage{array}
\usepackage{array}
\usepackage{units}
\usepackage{wrapfig}
\usepackage{pgfplots}
\usepackage{color}
\usepackage{multirow}
\usepackage{needspace}
\usepackage{rotating}
\usepackage{graphicx}
\usepackage{microtype}
\usepackage[margin=1in]{geometry}
\graphicspath{ {figures/} }

\title{Continuous-Flow Graph Transportation Distances}

\newcommand{\Prob}[0]{\textrm{Prob}}
\newcommand{\W}[0]{\mathcal{W}}
\newcommand{\R}[0]{\mathbb{R}}

\newtheorem{proposition}{Proposition}

\makeatletter
\renewenvironment{proof}[1][\proofname]{\par
  \vspace{-\topsep}
  \pushQED{\qed}%
  \normalfont
  \topsep0pt \partopsep-3pt 
  \trivlist
  \item[\hskip\labelsep
        \itshape
    #1\@addpunct{.}]\ignorespaces
}{%
  \popQED\endtrivlist\@endpefalse
  \addvspace{6pt plus 6pt} 
}
\makeatother

\author{
Justin Solomon\thanks{\url{jsolomon@mit.edu}} \\ MIT
\and
Raif Rustamov \\ AT\&T Labs 
\and
Leonidas Guibas \\ Stanford University	
\and
Adrian Butscher \\ Autodesk Research
}

\date{}

\begin{document}

\maketitle

\begin{abstract}
Optimal transportation distances are valuable for comparing and analyzing probability distributions, but larger-scale computational techniques for the theoretically favorable quadratic case are limited to smooth domains or regularized approximations.  Motivated by fluid flow-based transportation on $\mathbb{R}^n$, however, this paper introduces an alternative definition of optimal transportation between distributions over graph vertices.  This new distance still satisfies the triangle inequality but has better scaling and a connection to continuous theories of transportation.  It is constructed by adapting a Riemannian structure over probability distributions to the graph case, providing transportation distances as shortest-paths in probability space.  After defining and analyzing theoretical properties of our new distance, we provide a time discretization as well as experiments verifying its effectiveness.
\end{abstract}

\section{Introduction}

A fundamental problem affecting many application areas is the processing of signals on graphs~\cite{shuman-2013}.  In addition to representing networks, graphs provide an underlying structure for data in which connectivity and proximity affect processing.  Distances as measured along graph edges induce notions of continuity and displacement, allowing for the adaptation of ideas from continuous geometry to sampled or fundamentally discrete data.

In particular, we focus on probability distributions \emph{over} graphs, whereby each vertex of a graph is associated with some amount of probabilistic mass. Given that countless algorithms and models in machine learning rely upon a notion of distance/divergence between probability distributions, it is desirable to design efficiently-computable and well-behaved distances taking into account the connectivity of the underlying graph.

A well-known notion of distance from information theory is provided by the KL divergence and its variants.  The crucial drawback of these divergences in our context, however, is that they are agnostic to graph connectivity. That is, displacing probability from one node to another has the same cost regardless of the distance it travels along the edges.

One alternative taking connectivity into account is provided by \emph{optimal transportation} distances, also known as \emph{earth mover's distances} (EMDs).  These measure the work needed to transform one probability distribution into another by transporting mass over the graph edges~\cite{villani-2003}. Setting the cost of moving a unit of mass between graph nodes $v$ and $w$ to $d(v,w)^p$ yields the $p$-Wasserstein transportation distance, where $d(\cdot,\cdot)$ denotes shortest-path graph distance.

Transportation distances are the topic of a well-developed mathematical theory establishing geometric properties for the $2$-Wasserstein distance desirable for applications. Unfortunately, the basic optimization for this distance involves a linear program scaling quadratically with the number of nodes, rendering the computation prohibitively expensive even for moderately-sized graphs. This has led to approximations using wavelets~\cite{cohen-2011}, smoothing~\cite{cuturi-2013}, sketching~\cite{mcgregor-2013}, and so on. These approximations impose strict requirements on the graph and can become noisy or discontinuous thanks to aggressive approximation.  Others resort to the $1$-Wasserstein distance, which is easier computationally but has weak stability and geometric structure.

Here, we introduce a novel transportation distance between distributions over graphs. Our approach is inspired by a Riemannian framework for Wasserstein distance in $\R^n$ based on fluid flow~\cite{lott-2008}.  We consider the set of distributions on a graph as a manifold and define appropriate analogs of tangent spaces and inner products; our transportation distances are geodesics in this abstract space.  Although it does not come from the usual transportation linear program, our alternative definition still satisfies the triangle inequality and so on while benefiting from the new infinitesimal construction.

Our approach has several advantages. It scales linearly with the number of edges, enabling larger-scale computation especially on sparse graphs. Even with this favorable scaling, it has properties in common with $2$-Wasserstein distances.  The Riemannian structure, not known to accompany the quadratically-scaling linear programming formulation, gives our distances additional properties in common with Wasserstein distances over $\R^n$, such as ``displacement interpolation'' paths explaining the motion of mass from one distribution to another.  We also provide a suitable discretization and experiments suggesting our distance's structure and applicability to learning tasks.

\section{Optimal Transportation on Continuous Domains}

Suppose $M$ is a manifold with geodesic distance $d(\cdot,\cdot).$  We will use $\Prob(M)$ to denote the space of probability measures over $M$. Then, the $p$-\emph{Wasserstein} distance between $\mu_0,\mu_1\in\Prob(M)$ is:
\begin{equation}\label{eq:wasserstein}
\W_p(\mu_0,\mu_1)\equiv\inf_{\pi\in\Pi(\mu_0,\mu_1)}\left(
\iint_{M\times M} d(x,y)^p\,d\pi(x,y)
\right)^{\nicefrac{1}{p}}
\end{equation}
Measures $\pi\in\Pi(\mu_0,\mu_1)\subseteq\Prob(M\times M)$ are \emph{transportation plans} satisfying $\pi(U\times M)=\mu_0(U)$ and $\pi(M\times V)=\mu_1(V)$ for $U,V\subseteq M$.  We can think of $\pi\in\Pi(\mu_0,\mu_1)$ as a matching of probabilistic mass from $\mu_0$ to $\mu_1$; then, $\W_p$ is the minimum cost matching, assuming the cost of moving mass from $x\in M$ to $y\in M$ is $d(x,y)^p$.  We refer the reader to~\cite{villani-2003} for additional discussion.

In vision and learning, $\W_1$ is known as the ``earth mover's distance''~\cite{rubner-2000}, used on histograms and binned descriptors.  Theoretically, however, the quadratic distance $\W_2$ admits stronger smoothness and uniqueness properties akin to the difference between $L_2$ and $L_1$ regularization.  Regardless, algorithms optimizing~\eqref{eq:wasserstein} directly deal with distributions over the cross product $M\times M$; discretizing $M$ using $n$ points produces a linear program over $O(n^2)$ variables, which can be prohibitively expensive.

More recent methods establish a connection between~\eqref{eq:wasserstein} and Eulerian fluid mechanics. Most prominently, Benamou and Brenier show that $\W_2$ can be computed on $M$ as follows~\cite{benamou-2000}:
\begin{equation}\label{eq:bb}
[\W_2(\rho_0,\rho_1)]^2=
\left\{
\begin{aligned}
\inf_{\rho(t,x),v(t,x)} & \int_M\int_0^1 \rho(t,x) \|v(t,x)\|^2\,dx\,dt\\[1ex]
\mbox{such that } \; & \rho(0,x)=\rho_0(x),\
\rho(1,x) = \rho_1(x)\\
& \frac{\partial \rho}{\partial t} + \nabla\cdot(\rho v) = 0,\ \rho(t,x)\geq0
\end{aligned}
\right.
\end{equation}
We assume measures $\mu$ can be expressed using distributions $\rho$ such that $\mu(U)\equiv\int_U \rho(x)\,dx$.  The variable $\rho(t,x)$ is a time-varying set of distributions advecting from $\rho_0$ to $\rho_1$ along vector field $v(t,x)$ with minimal kinetic energy $\iint\rho \|v\|^2 \, dx \, dt$.  Discretizations of this problem scale well, since $[0,1]$ typically can be subdivided into far fewer than $n$ subintervals.  Similar formulations exist for $\W_1$~\cite{santambrogio-2013}, but we focus on $\W_2$ for its structure and computational challenges on graphs.

Using~\eqref{eq:bb}, $\W_2$ can be viewed as a geodesic distance on $\Prob(M)$~\cite{lott-2008}.  First-order optimality of~\eqref{eq:bb} shows $v(t,x)=\nabla \phi(t,x)$ for some $\phi(t,x)$.  This is interpreted to mean that the set of gradients $\{\nabla \phi(\cdot)\}$ forms a ``tangent space" for $\Prob(M)$ at $\mu$ equipped with an inner product defined by
\begin{equation}\label{eq:inner_prod_bb}
\langle \nabla\phi_1,\nabla\phi_2\rangle_{\mu}\equiv\int_M \langle \nabla\phi_1(x),\nabla\phi_2(x)\rangle\,d\mu(x).
\end{equation}
Given a time-varying set of distributions $\rho(t,x)$, there exists a unique potential field $\nabla_x \phi(t,x)$ with $\frac{\partial\rho}{\partial t}+\nabla\cdot(\rho\nabla \phi)=0$, so if we write a ``curve'' $c(t):[0,1]\rightarrow\Prob(M)$ corresponding to $\rho(t,x)$, then we can define $c'(t)\equiv \nabla\phi(t,\cdot).$  Then,~\cite{lott-2008} proves:
\begin{equation}\label{eq:lott}
\W_2(\rho_0,\rho_1) = \inf_{\substack{c(0)=\rho_0\\c(1)=\rho_1}} \int_0^1 \langle c'(t),c'(t)\rangle_\mu^{\nicefrac{1}{2}}\,dt.
\end{equation}
That is, quadratic Wasserstein distance satisfies the geodesic equation under this inner product.

\section{Continuous-Flow Transportation Distances}

Now, we transition from manifold domains $M$ to undirected connected graphs $G=(V,E).$  Recall that our goal is to construct a transportation distance on discrete distributions using the structure of $G$ with the structure of $2$-Wasserstein distances in Euclidean space.

Define $d(v,w)$ to be the shortest-path distance between $v,w\in V$, and define $\Prob(G)\equiv\{p\in[0,1]^{|V|}:\mathbbm1^\top p=1\}.$  To mimic the construction of $\W_2$, we might attempt to compute the following quadratic transportation distance on $\Prob(G)$:
\begin{equation}\label{eq:pairwise}
[\overline\W_{\mathrm{full}}(p_0,p_1)]^2\equiv\left\{
\begin{aligned}
\inf_{T(v,w)\geq0} & \sum_{v,w\in V} d(v,w)^2T(v,w)\\
\mbox{such that } \; 
& \sum_{w\in V} T(v,w)=p_0(v)\ \forall v\in V\\
& \sum_{v\in V} T(v,w)=p_1(w)\ \forall w\in V
\end{aligned}\right.
\end{equation}
While this distance has some properties in common with $\W_2$, the linear program requires $|V|^2$ variables $T(v,w)$ and precomputation of a dense matrix of pairwise distances $d(v,w).$

To alleviate these scaling problems, we might hope to find an Eulerian formulation of $\overline\W_{\mathrm{full}}$ in the style of~\eqref{eq:bb}.  Such a distance is challenging to derive starting from~\eqref{eq:pairwise} for several reasons.  First, graphs are discrete with no perfect analog of the advection equation $\frac{\partial\rho}{\partial t}+\nabla\cdot(\rho v)=0$.  Furthermore, flows on graphs are per-edge and distributions are per-vertex, so the objective $\rho\|v\|^2$ must be approximated.  Hence, rather than expecting to find a new formulation of $\overline\W_{\mathrm{full}},$ we will introduce a \emph{new} transportation distance constructed by working backward from~\eqref{eq:lott}.

\subsection{Definition}

The main ingredients defining transportation distances starting from~\eqref{eq:lott} are:
\begin{enumerate}
\item A model of advection for moving probabilistic mass between adjacent vertices of a graph over time using a per-edge flow.
\item A probability-based inner product on per-edge flows similar to~\eqref{eq:inner_prod_bb}.
\end{enumerate}

To simplify notation, define $\overline E\equiv \{(v\rightarrow w), (w\rightarrow v): (v,w)\in E\},$ a set of directed edges where each edge in $E$ is doubled with forward and backward orientation.  Following~\cite{chapman-2011}, given a nonnegative oriented flow $U(t,e):[0,1]\times\overline E\rightarrow \R^+$ and distribution $p_0\in\Prob(G)$, we define advection of $p_0$ along $U$ as the solution of the system of ordinary differential equations
\begin{align}\label{eq:graph_advection}
\frac{d}{dt}p(t,v) &= \sum_{e=(w\rightarrow v)} U(t,e) p(t,w) - \sum_{e=(v\rightarrow w)} U(t,e) p(t,v)\,,\\
p(0,v) &= p_0(v)\,.\nonumber
\end{align}
The time variable $t$ aside, this model of advection is defined discretely in terms of the structure of the graph but still preserves mass in that $\sum_{v\in V} p(t,v)=1$ for all $t\geq0$.

To define an inner product on flows, take $U,W:\overline E\rightarrow\R$ and $p\in\Prob(G)$ with $p(v)>0\ \forall v\in V$.  Then, we define the \emph{advective inner product} between $U$ and $W$ at $p$ as:
\begin{equation}\label{eq:advective_inner_prod}
\langle U,W\rangle_p\equiv\sum_{e=(v\rightarrow w)} \left(\frac{p(v)}{p(w)}\cdot\frac{(p(v)+p(w))}{2}\right) U(e)W(e)\,.
\end{equation}
This is an $L_2$ inner product on $\R^{2|E|}$ weighted by a function of the distribution $p$.  The asymmetry between $p(v)$ and $p(w)$ accounts for the structure of graph advection~\eqref{eq:graph_advection} and makes our transportation distance symmetric via Proposition~\ref{prop:momentum}.  We also define the \emph{advective norm} $\|U\|_p\equiv\sqrt{\langle U,U\rangle_p}.$

Now, we can define the \emph{continuous-flow transportation distance} on $\Prob(G)$ by imitating~\eqref{eq:lott}:
\begin{empheq}[box=\fbox]{equation}\label{eq:graph_geodesic}
\overline\W(p_0,p_1)\equiv\left\{
\begin{aligned}
\inf_{\substack{U(t,e)\geq0\\p(t,v)\geq0}} & \int_0^1 \|U(t,\cdot)\|_{p(t,\cdot)}\,dt\\
\mbox{such that }\;  & p(0,v)=p_0(v), \, p(1,v)=p_1(v)\\
&\text{(\ref{eq:graph_advection}) holds}
\end{aligned}
\right.
\end{empheq}

\subsection{Properties}

We state several properties of our new distance $\overline\W$ to provide intuition for its behavior and to motivate how we will compute it numerically.  We begin by providing an alternative convex formulation:

\begin{proposition} $\overline\W$ can be computed as follows:\label{prop:momentum}
\begin{equation}\label{eq:momentum_form}
[\overline\W(p_0,p_1)]^2=\left\{
\begin{aligned}
\inf_{\substack{J(t,e)\geq0\\p(t,v)\geq0}} & \int_0^1\sum_{e=(v\rightarrow w)} \frac{J(t,e)^2}{2}\left(\frac{1}{p(t,v)} + \frac{1}{p(t,w)}\right)\,dt\\
\mbox{such that } \; & p(0,v)=p_0(v), \, p(1,v)=p_1(v)\ \forall v\in V\\
& \frac{d}{dt} p(t,\cdot) = D^\top J(t,\cdot)
\end{aligned}
\right.
\end{equation}
Here, $D\in\R^{2|E|\times |V|}$ is the operator computing $p(w)-p(v)$ for each oriented edge $(v\rightarrow w)\in \overline E.$
\end{proposition}
\begin{proof}
First, substitute $J(t,e)\equiv p(t,v)U(t,e)$ to~\eqref{eq:graph_geodesic}, for $e=(v\rightarrow w).$  Then,
$$\langle U,U\rangle_p
=\sum_{e=(v\rightarrow w)\in\overline E} \frac{J(t,e)^2}{2}\left(\frac{1}{p(t,v)} + \frac{1}{p(t,w)}\right)\,.$$
An identical substitution into the advection equation shows that $J$ satisfies $\frac{dp}{dt}=D^\top J.$

Define a new inner product $\langle\cdot,\cdot\rangle_p'$ replacing the weights in~\eqref{eq:advective_inner_prod} with $\frac{1}{p(v)}+\frac{1}{p(w)}$. 
In this notation, we must show that minimizing~\eqref{eq:momentum_form}, whose energy functional now can be written as $\int_0^1 \|J(t,\cdot)\|_{p(t,\cdot)}'^2\,dt$, is equivalent to squaring the result of minimizing the non-squared functional $\int_0^1 \|J(t,\cdot)\|_{p(t,\cdot)}'\,dt$ with the same constraints.  We follow an analogous proof for manifold geodesics in Theorem 2.1 of~\cite{udriste-1994}.

First, take any $J$ and $p$ satisfying the constraints.  As a consequence of H\"older's inequality, we know
$$\int_0^1 \|J(t,\cdot)\|_{p(t,\cdot)}'\,dt\leq\left(\int_0^1 \|J(t,\cdot)\|_{p(t,\cdot)}'^2\,dt\right)^{\nicefrac{1}{2}}.$$
This establishes~\eqref{eq:momentum_form} as an upper bound for~\eqref{eq:graph_geodesic}.

Now, suppose $(J,p)$ are minimizers of~\eqref{eq:graph_geodesic} after substitution.  Define the ``arc length'' function
$\ell(t)\equiv \int_0^t \|J(\bar t,\cdot)\|_{p(\bar t,\cdot)}'\,d\bar t,$
and take $s(t)\equiv\nicefrac{\ell(t)}{\ell(1)}.$  
Reparameterizing $p$ with respect to $s$ yields
$$\frac{d}{ds}p(s,\cdot)=\frac{dp}{dt}\frac{dt}{ds}=D^\top J(s,\cdot)\frac{dt}{ds}=D^\top\left(\frac{J(s,\cdot)\ell(1)}{\|J(s,\cdot)\|_p}\right)\equiv D^\top\tilde J(s,\cdot).$$
Hence, $\tilde J(s,\cdot)$ and $p(s,\cdot)$ satisfy the constraints with $t\mapsto s$.  The advective norm of $\tilde J$, however, is a constant function of $s$ with $\|\tilde J(s,\cdot)\|'_p = \ell(1) = \W(p_0,p_1).$   Thus we know $\overline\W(p_0,p_1)^2=\int_0^1 \|\tilde J(s,\cdot)\|'^2_p\,ds,$ establishing~\eqref{eq:graph_geodesic} as an upper bound for~\eqref{eq:momentum_form}.\end{proof}

With this alternative formulation in hand, we verify that $\overline\W$ truly satisfies the properties of a distance metric on $\Prob(G)$:
\begin{proposition} $\overline\W$ is a distance on $\Prob(G)$.
\end{proposition}\begin{proof}
Non-negativity follows directly from the form of $\overline\W$.  Symmetry follows from~\eqref{eq:momentum_form} since $\overline E$ contains all forward \emph{and} backward edges; any forward-flowing $p(t,v)$ can be replaced by a backward flow $p(1-t,v)$ with the same integrated advective norm by placing $J$ on flipped edges.  Finally, the triangle inequality follows from~\eqref{eq:graph_geodesic} identically to the proof of triangle inequality for geodesic curves by concatenating paths in $\Prob(G)$.
\end{proof}

Finally, we relate $\overline\W$ to the geometry of the underlying graph.  For simplicity, we restrict our theoretical result to a weak but easily-proved bound and show experimentally in \S\ref{sec:experiments} that our bound is conservative:

\begin{proposition} Suppose $\delta_v,\delta_w\in\Prob(G)$ are the indicator functions of $v,w\in V$.  Then, there exists $c\geq0$ not dependent on $G$ with $\overline\W(\delta_v,\delta_w)\leq cd(v,w).$\label{prop:bound}
\end{proposition}
\begin{proof}
For a two-node graph with $V=\{a,b\}$ and $E=\{(a,b)\}$, define $c\equiv\overline\W(\delta_a,\delta_b)$; using the discretization in \S\ref{sec:discretization} with a large number of time steps, $c\approx 2.2159$.  Then, for $v,w\in V$ in a general graph $G=(V,W)$, define $\ell\equiv d(v,w)$.  We construct a time-varying set of distributions $p(t,v)$ by dividing $[0,1]$ into $\ell$ subintervals of length $\nicefrac{1}{\ell}$ and use two-node solution to transport mass along each edge of the shortest path from $v$ to $w$ one-at-a-time.  Compressing the two-node integral into an interval of length $\nicefrac{1}{\ell}$ scales the integral for $\overline\W^2$ by $\ell$, and there are $\ell$ such subintervals.  Thus, this feasible point for~\eqref{eq:momentum_form} shows $[\overline\W(\delta_v,\delta_w)]^2\leq c^2\ell^2\implies \overline\W(\delta_v,\delta_w)\leq c\ell.$
\end{proof}

\subsection{Discretization}\label{sec:discretization}

While $\overline\W$ is a distance on $\Prob(G)$, a subset of the finite-dimensional set $[0,1]^{|V|}$, its computation involves solving variational problems~\eqref{eq:graph_geodesic} or~\eqref{eq:momentum_form} containing a continuous time variable $t$.   For this reason, we propose a discrete approximation of $\overline\W$ by sampling $t$ before the minimization.  This approximation can be evaluated using standard convex optimization machinery.

Since the functional $\nicefrac{x^2}{y}$ is convex so long as $y\geq0$, we use~\eqref{eq:momentum_form} as the starting point for our discretization; this objective term can be included in second-order cone programs~\cite{grant-2008,cvx-2014}.  We divide the time span $[0,1]$ into $k$ subintervals, with improving approximation quality as $k\rightarrow\infty$.  Similar to~\cite{papadakis-2014}, we propose using a \emph{staggered} discretization in which $J$ is sampled at subinterval midpoints and $p$ is sampled at subinterval endpoints; for this reason, our optimization variables are distributions $q_0,\ldots,q_k\in\Prob(G)$ with flows $J_1,\ldots,J_k$ so that $J_i$ is the flow from $q_{i-1}$ to $q_i$.

Our discretized version of the optimization~\eqref{eq:momentum_form} is: 
\begin{equation}\label{eq:discretization}
[\overline\W_k(p_0,p_1)]^2\equiv
\left\{
\begin{aligned}
\min_{\substack{J^i(e)\geq0\\q^i(v)\geq0}}\;& k\sum_{i=1}^k \sum_{e=(v\rightarrow w)} \frac{J^i(e)^2}{2}\left(\frac{1}{q^{i-1}(v)}+\frac{1}{q^i(w)}\right)\,,\\
\mbox{such that } \; & q^0 = p_0, \, q^k = p_1, \,
 D^\top J^i = q^i - q^{i-1}\,.
\end{aligned}
\right.
\end{equation}
This discretization is convex. Taking $J=1$ along the shortest path from $v$ to $w$, one can show $\overline\W_k(\delta_v,\delta_w)=k$ when $d(v,w)=k,$ suggesting that the bound in Proposition~\ref{prop:bound} is conservative.

One feature of our discretization is the mix of $q^{i-1}(v)$ and $q^i(w)$ in the optimization objective.  This construction ensures that when $J^i(e)>0$, both denominators are nonzero since mass is transported from vertex $v$ at time $i-1$ to vertex $w$ at time $i$.  This slight asymmetry allows the boundary distributions to contain zero values and vanishes for large $k$.  If symmetry is desired, the forward and backward discretizations can be averaged.

This optimization problem has $O(k(|V|+|E|))$ variables.  One advantage of such a formulation over the $O(|V|^2)$ linear programming approach is that additional edge-based pruning strategies can be applied to reducing the number of variables when we are confident certain edges will not figure into the computation.  For instance, we can first compute the one-Wasserstein distance using network flow techniques and prune those edges that are not used by the resulting flow.  This enhances efficiency when distributions are concentrated in the same region of the graph or have low entropy.

\section{Experiments}\label{sec:experiments}


\subsection{Convergence}

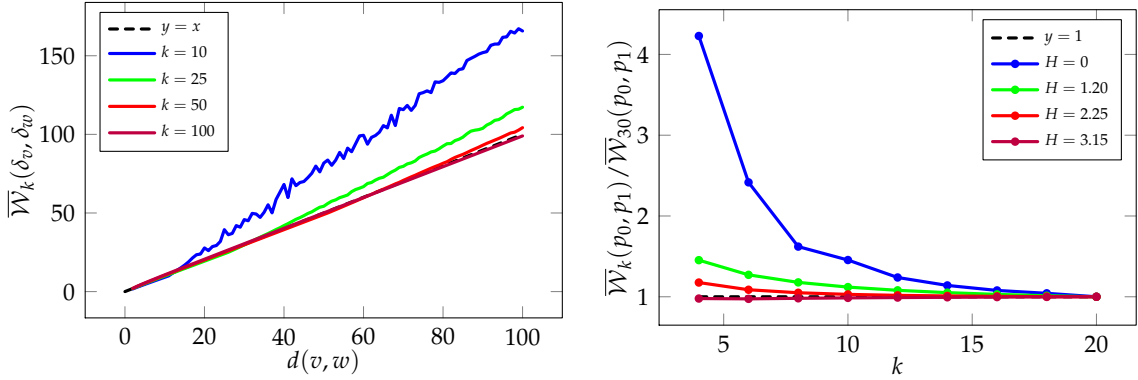
\begin{figure}[t]\centering
\begin{tabular}{cc}
\begin{tikzpicture}
\begin{axis}[
xlabel={\footnotesize $d(v,w)$},
ylabel={\footnotesize $\overline\W_k(\delta_v,\delta_w)$},
legend entries={$y=x$,$k=10$,$k=25$,$k=50$,$k=100$},
legend style={cells={anchor=west},legend pos=north west,font=\tiny},
width=.48\columnwidth,
height=.35\columnwidth,
every axis/.append style={font=\footnotesize},
xlabel near ticks,
ylabel near ticks,
line cap = round,
line join = round,
xlabel shift = -.08in,
ylabel shift = -.08in
]
\addplot[black,samples=300,very thick,dashed,domain=0:100]{x};
\addplot[mark=none,color=blue,very thick] table {figures/deltaconvergence/k10.dat};
\addplot[mark=none,color=green,very thick] table {figures/deltaconvergence/k25.dat};
\addplot[mark=none,color=red,very thick] table {figures/deltaconvergence/k50.dat};
\addplot[mark=none,color=purple,very thick] table {figures/deltaconvergence/k100.dat};
\end{axis}
\end{tikzpicture}
&
\begin{tikzpicture}
\tikzset{every mark/.append style={scale=.5}}
\begin{axis}[
xlabel={\footnotesize $k$},
ylabel={\footnotesize $\overline\W_k(p_0,p_1)/ \overline\W_{30}(p_0,p_1)$},
legend entries={$y=1$,$H=0$,$H=1.20$,$H=2.25$,$H=3.15$},
legend style={cells={anchor=west},legend pos=north east,font=\tiny},
width=.48\columnwidth,
height=.35\columnwidth,
every axis/.append style={font=\footnotesize},
xlabel near ticks,
ylabel near ticks,
line cap = round,
line join = round,
xlabel shift = -.08in,
ylabel shift = -.08in
]
\addplot[black,samples=300,very thick,dashed,domain=4:20]{1};
\addplot[mark=*,color=blue,very thick] table {figures/entropyconvergence/entropy0.dat};
\addplot[mark=*,color=green,very thick] table {figures/entropyconvergence/entropy1.11986.dat};
\addplot[mark=*,color=red,very thick] table {figures/entropyconvergence/entropy2.25171.dat};
\addplot[mark=*,color=purple,very thick] table {figures/entropyconvergence/entropy3.15293.dat};
\end{axis}
\end{tikzpicture}
\end{tabular}\vspace{-.15in}
\caption{Convergence of the discrete-time approximation.  (left) Graph distance between nodes $v$ and $w$ versus distance between indicator functions $\delta_v$ and $\delta_w$; (right) convergence of the approximation $\overline\W_k$ for different values of the average entropy $H$ between $p_0$ and $p_1$.}\label{fig:convergence}
\end{figure}

Figure~\ref{fig:convergence} illustrates the convergence of our approximation $\overline\W_k$ for increasing $k$.  
On the left, we measure $\overline\W_k$ between the indicators of two vertices $v,w\in V$ as in Proposition~\ref{prop:bound}.  We plot the relationship between graph distance and $\overline\W_k$ for various $k$.  As $k$ increases, $\overline\W_k(\delta_v,\delta_w)$ better approximates $d(v,w)$.  This suggests a commonality with traditional transportation distances, which satisfy the diagonal relationship exactly and shows that our bound in Proposition~\ref{prop:bound} likely could be tightened.
On the right, we plot $\overline\W_k$ as $k$ increases for assorted pairs $p_0,p_1\in\Prob(G)$.  In this experiment, we measure distances after perturbing the indicators of the two farthest vertices of a thirty-vertex line graph with uniform noise. We plot convergence for distances for different levels of entropy in $p_0$ and $p_1$; our approximation succeeds even with low $k$ when the entropies of the distributions is high.  This observation reflects the fact that transportation between high-entropy distributions is less likely to move mass large distances, since mass exists all over $G$.

\subsection{Displacement Interpolation}\label{sec:displacement_interpolation}

One well-understood phenomenon associated with the continuous two-Wasserstein distance $\W_2$ is \emph{displacement interpolation}, introduced by McCann in~\cite{mccann-1997}.  Displacement interpolation can be thought of as the characterization of the ``path'' of distributions $\rho(t,\cdot)$ from $\rho_0$ to $\rho_1$ suggested in the flow-based formulation~\eqref{eq:bb}.  This path is well-understood and isotropic in the $\W_2$ case, whereas for $\W_1$ the corresponding path is the trivial construction $\rho(t,x)=(1-t)\rho_0(x)+t\rho_1(x)$~\cite{villani-2003}.

Displacement interpolation largely is understood as a continuous phenomenon on connected rather than discrete domains.  Even in the quadratic case~\eqref{eq:pairwise}, there does not appear to be a well-characterized flow of probability over time explaining the optimal matching $T(\cdot,\cdot).$  Contrastingly, our distance $\overline\W$ exhibits displacement interpolation phenomena by construction.

\begin{figure}
\begin{tabular}{r@{}c@{}c@{}c@{}c@{}c}
\rotatebox{90}{\tiny\hspace{.09in}Trivial}&
\includegraphics[width=0.2\textwidth]{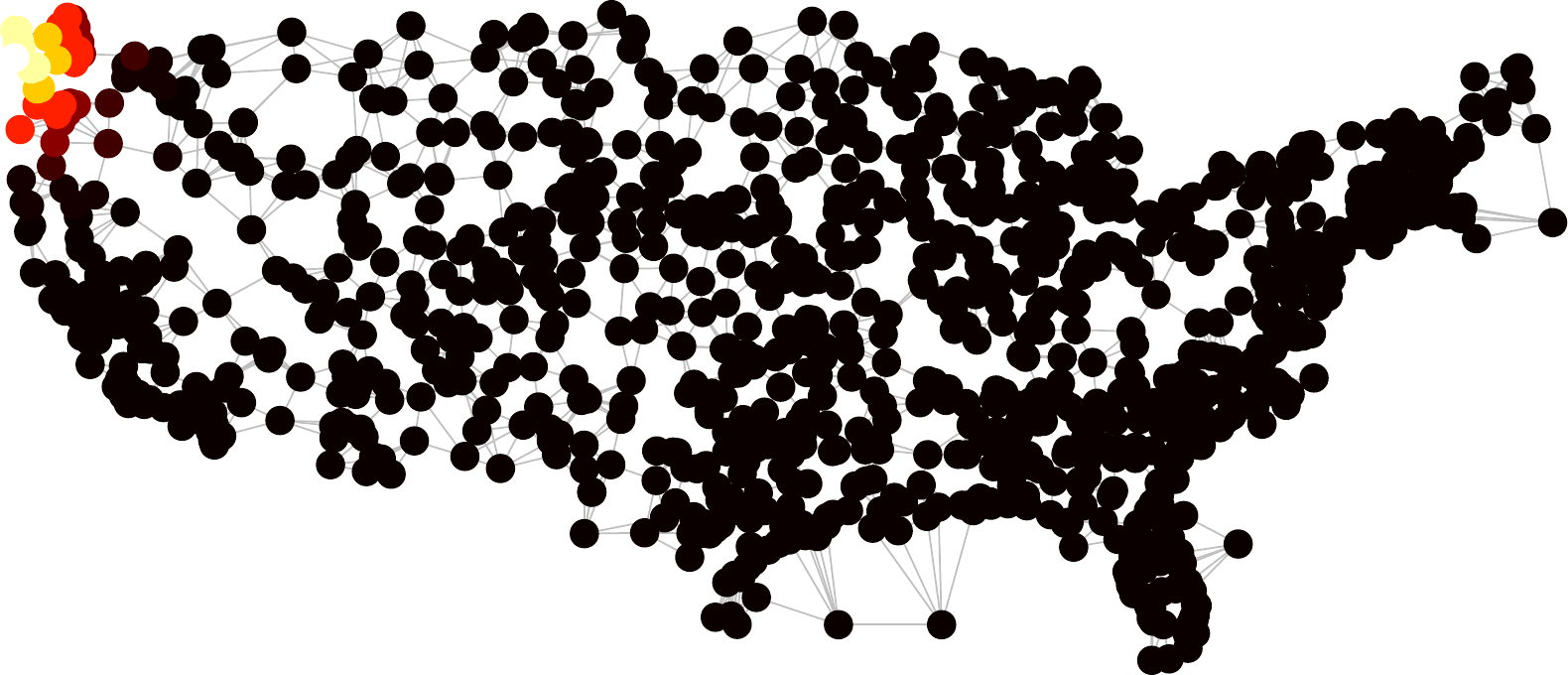}&
\includegraphics[width=0.2\textwidth]{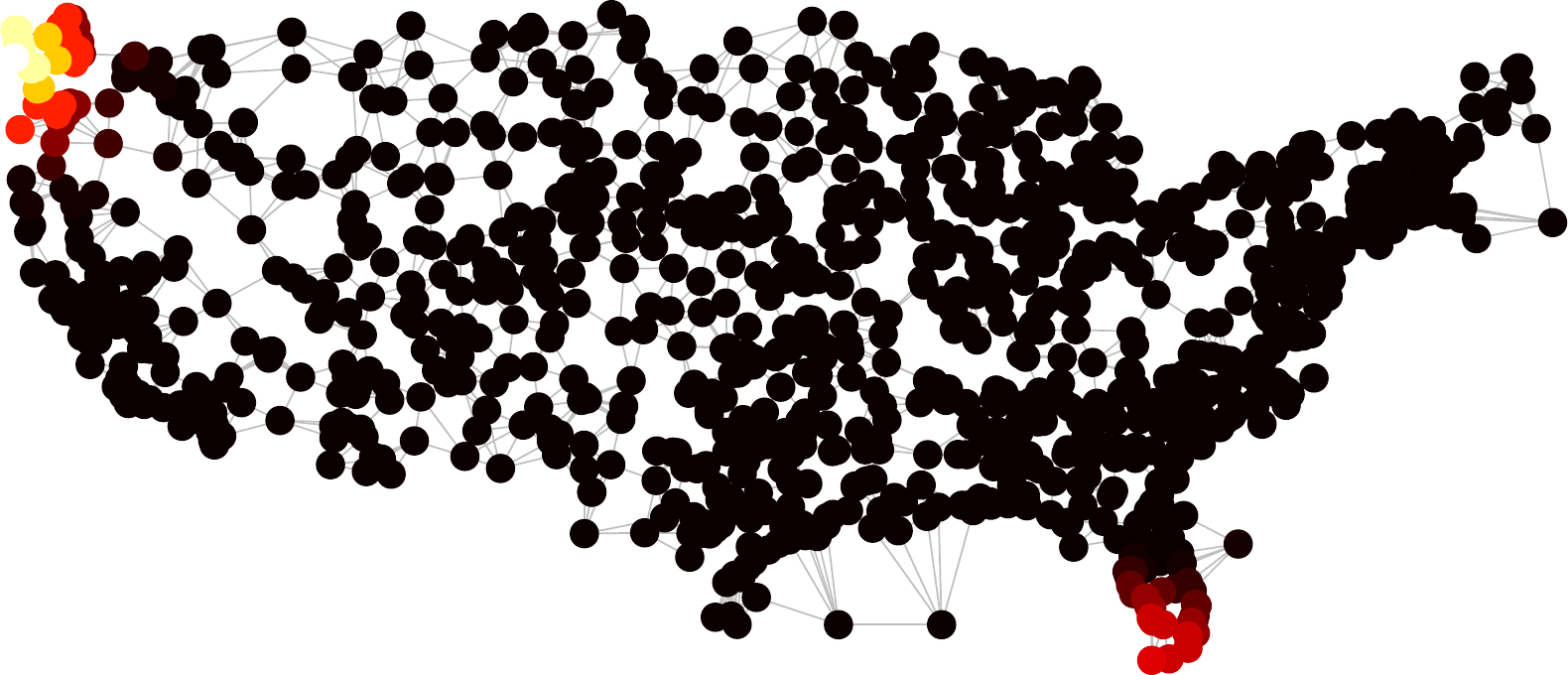}&
\includegraphics[width=0.2\textwidth]{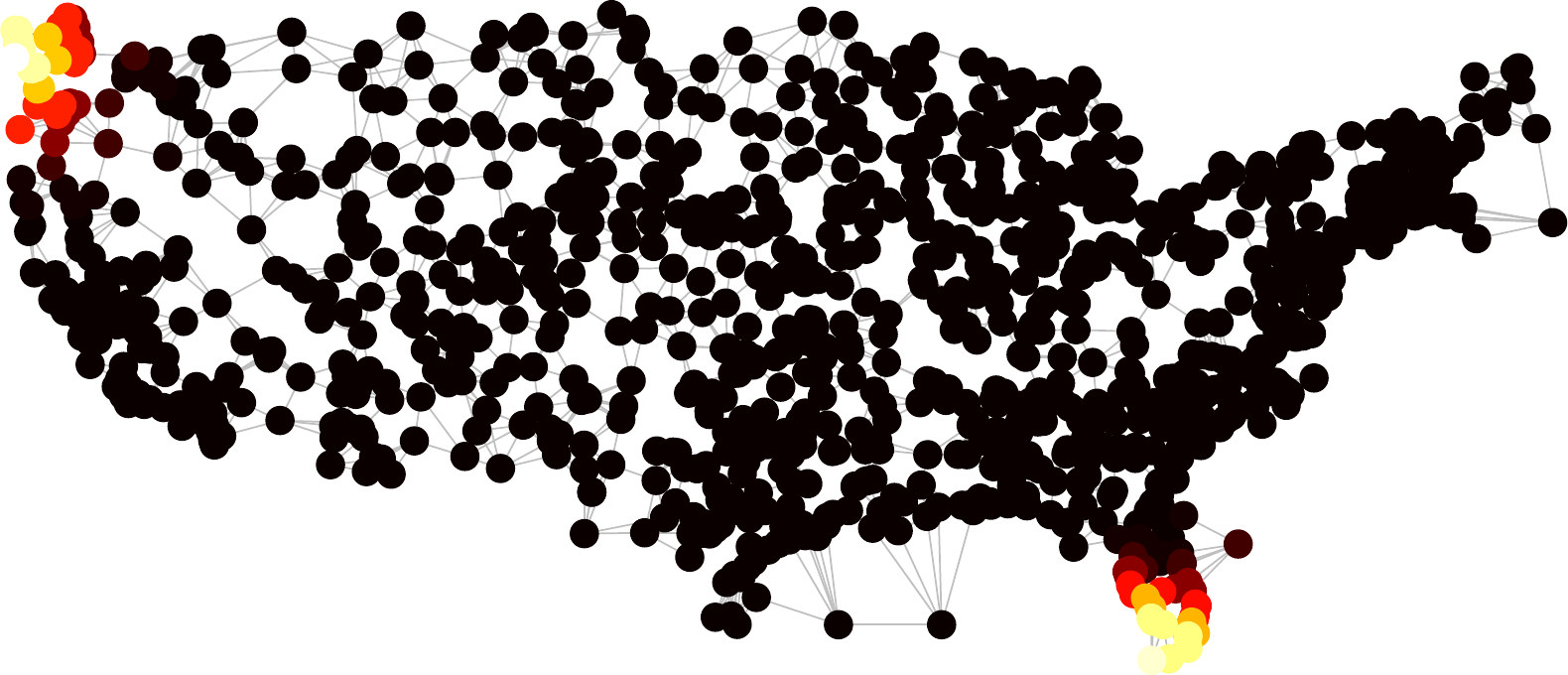}&
\includegraphics[width=0.2\textwidth]{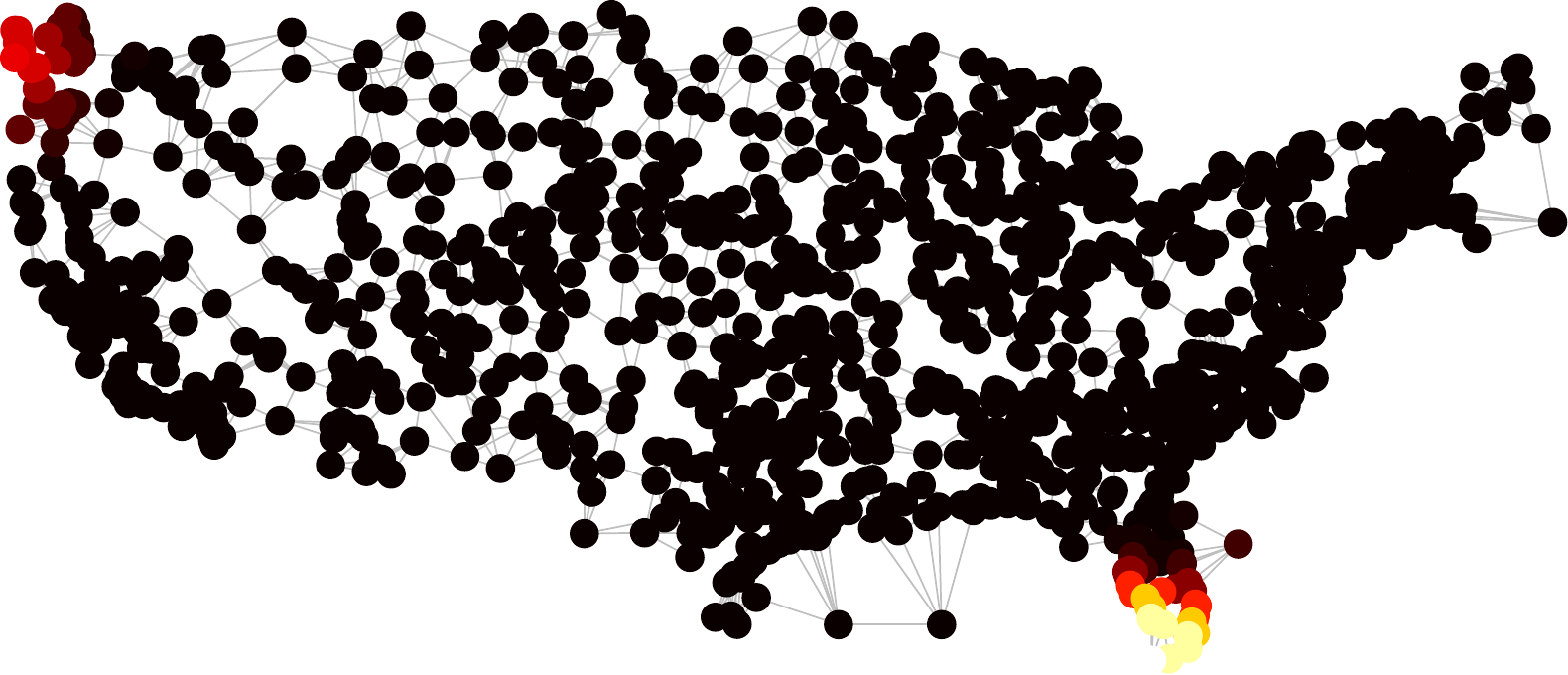}&
\includegraphics[width=0.2\textwidth]{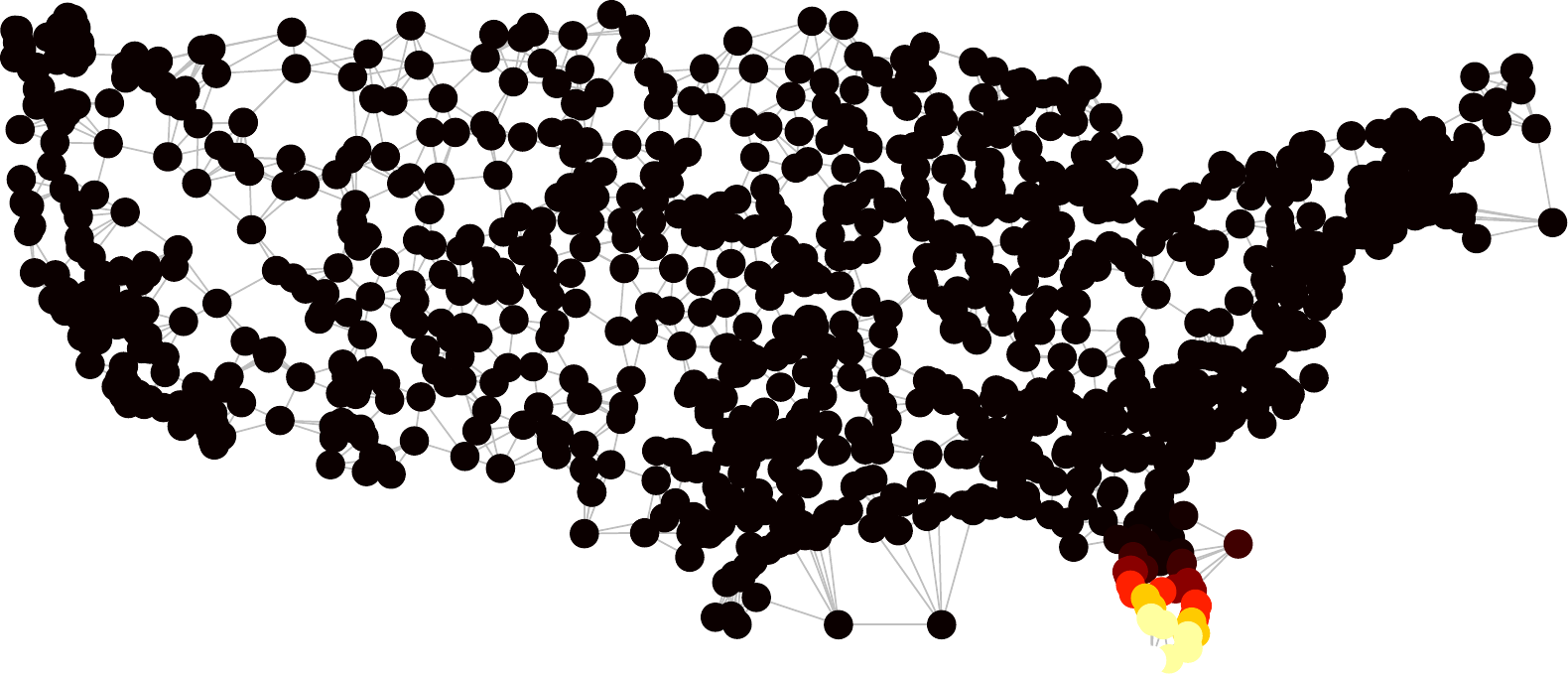} \\
\rotatebox{90}{\tiny \hspace{.15in}$\overline\W$}&
\includegraphics[width=0.2\textwidth]{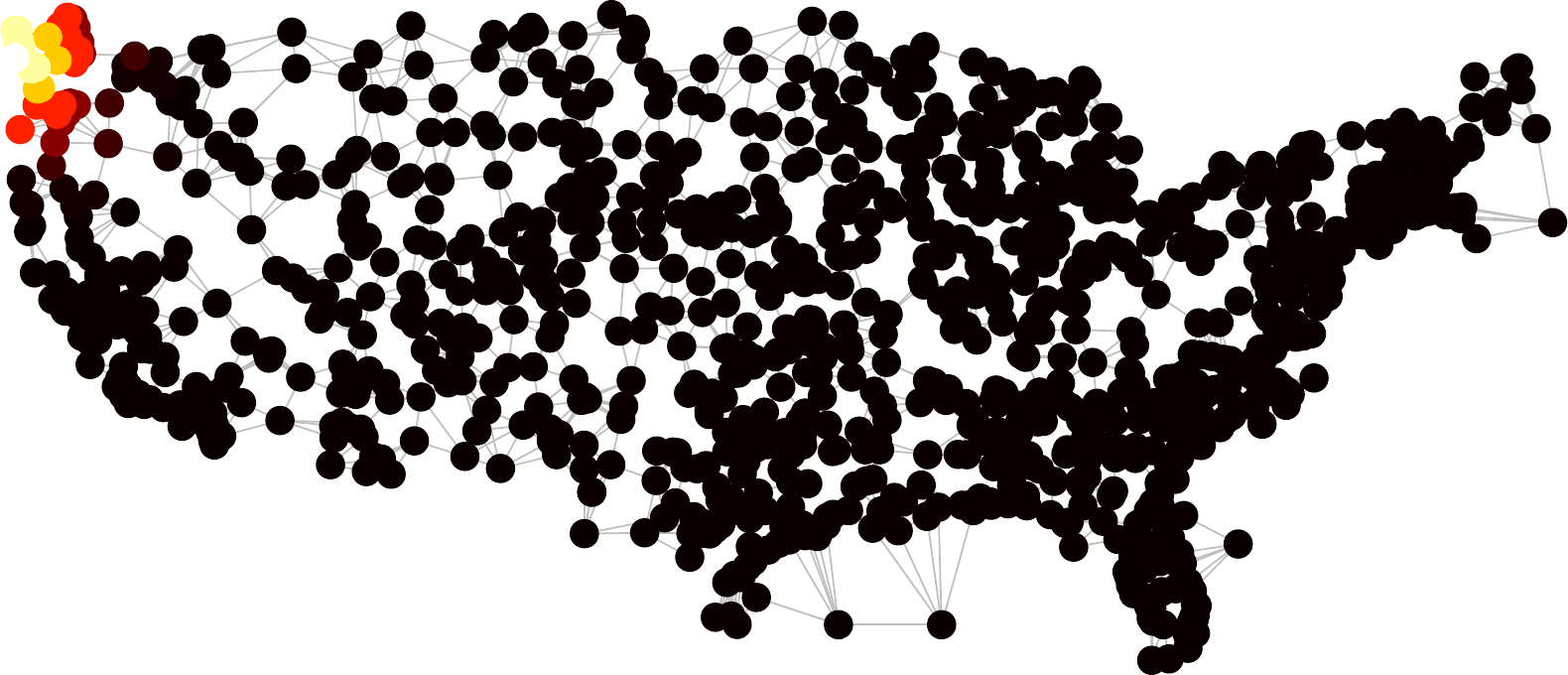}&
\includegraphics[width=0.2\textwidth]{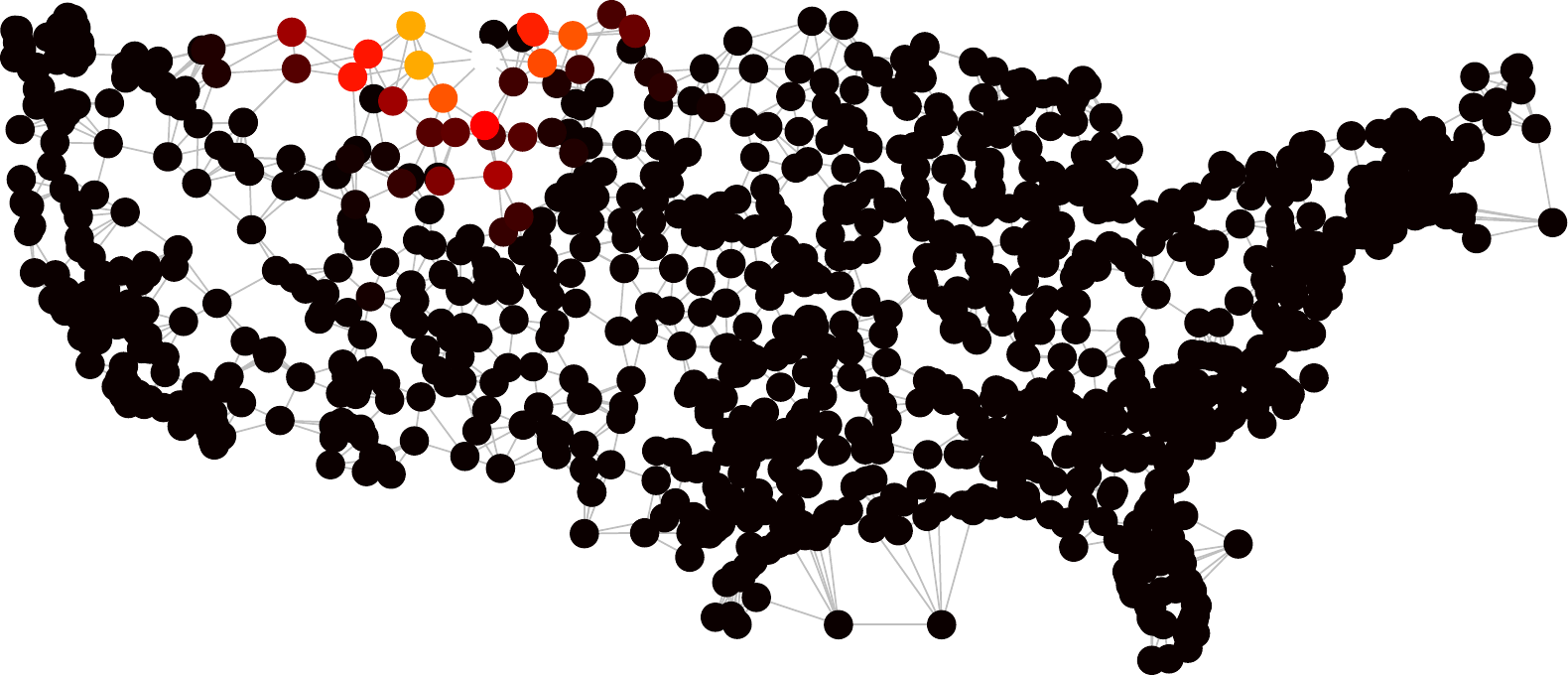}&
\includegraphics[width=0.2\textwidth]{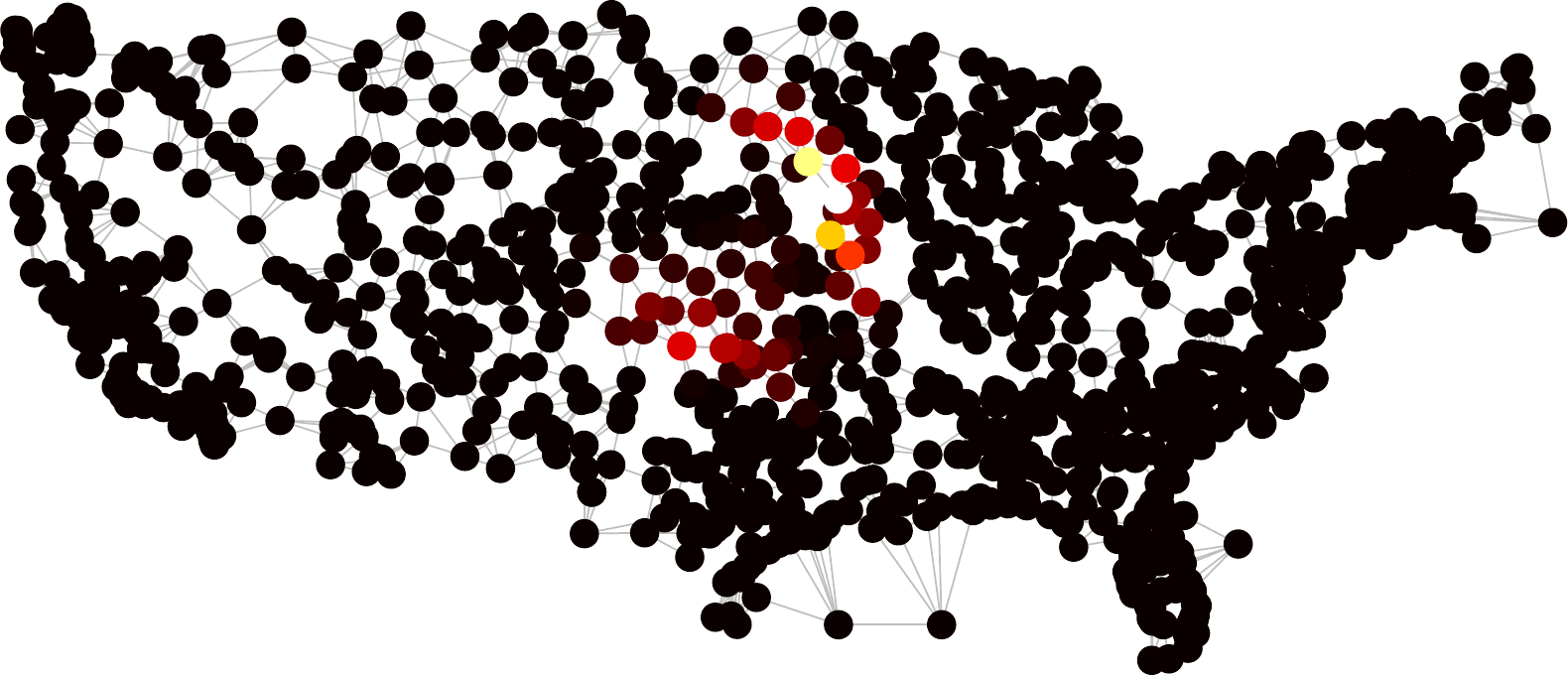}&
\includegraphics[width=0.2\textwidth]{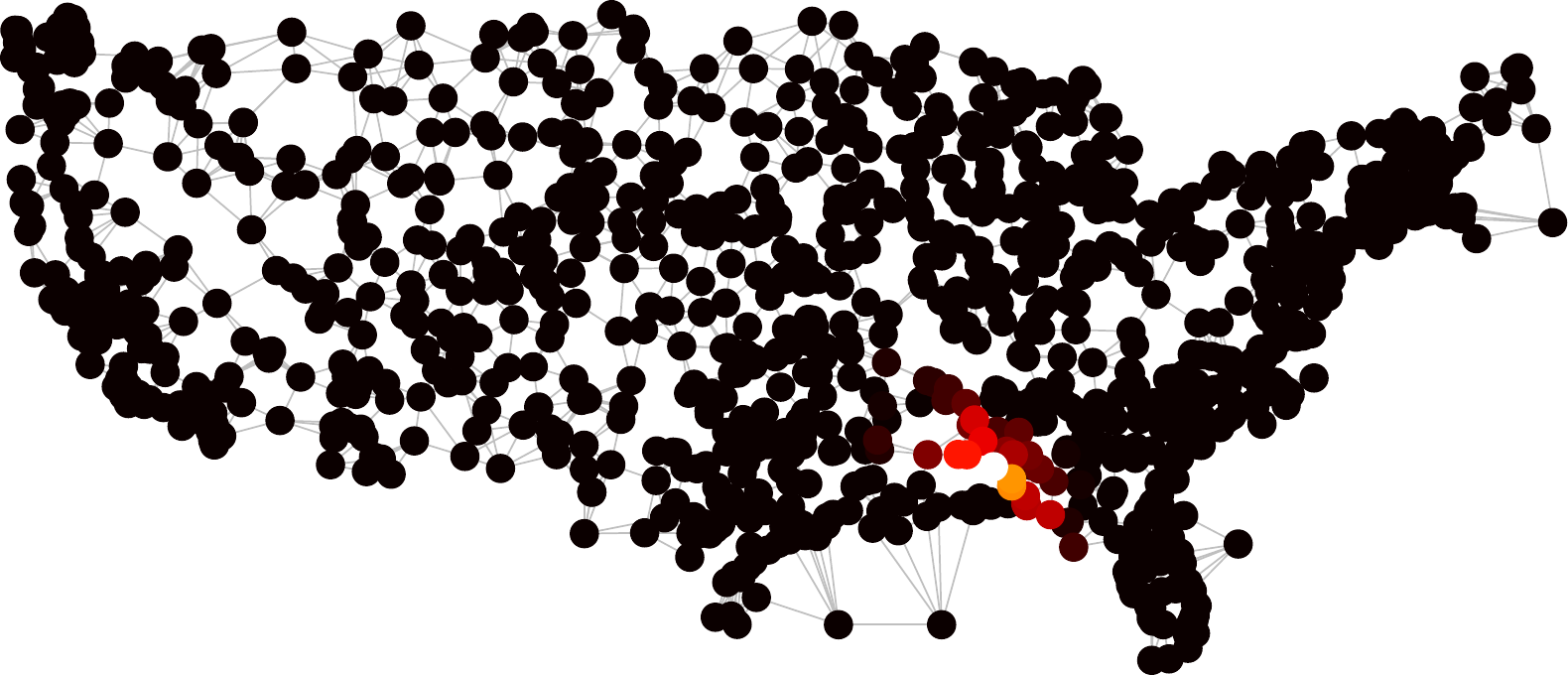}&
\includegraphics[width=0.2\textwidth]{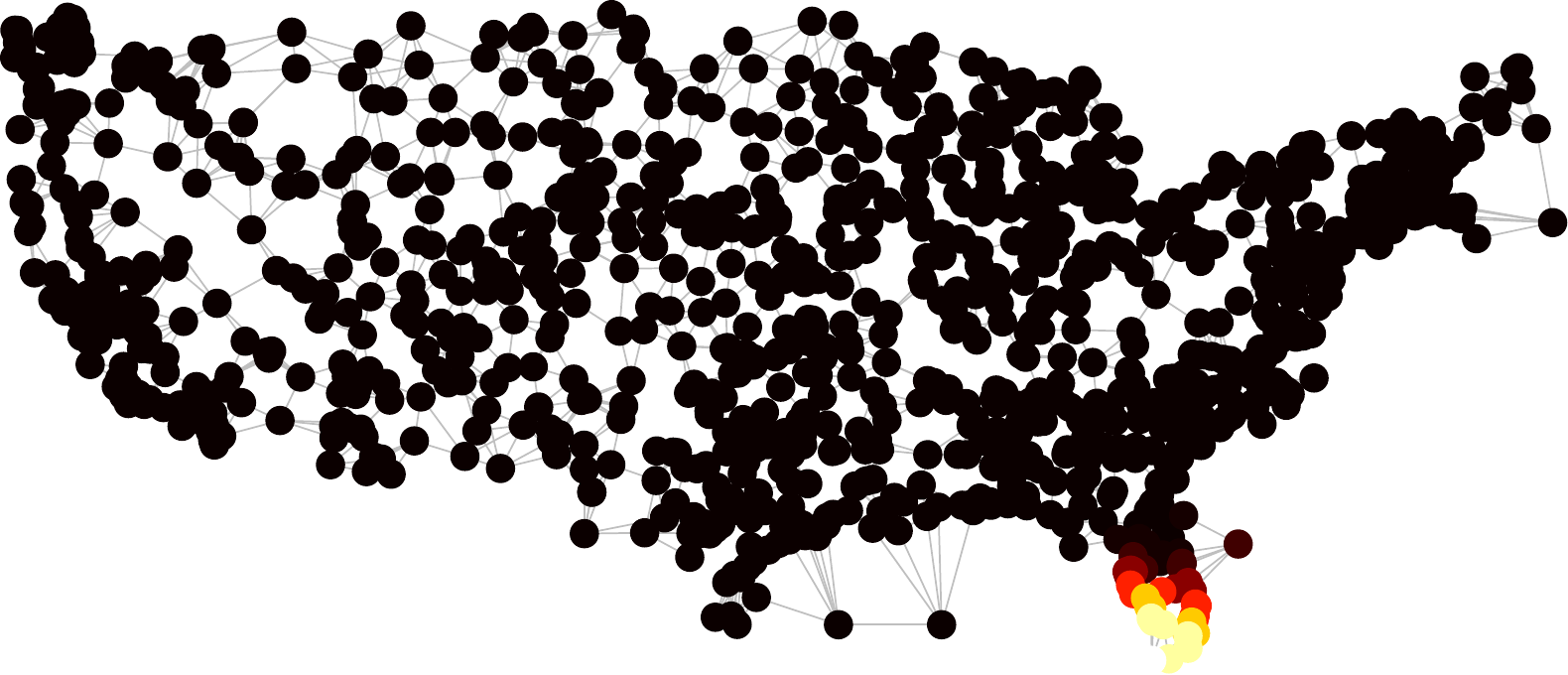} \\
& $t=0$ & $t=\frac{1}{4}$ & $t=\frac{1}{2}$ & $t=\frac{3}{4}$ & $t=1$
\end{tabular}
\vspace{-.1in}
\caption{Displacement interpolation on a geometric graph.  The one-Wasserstein distance $\W_1$ is explained by the trivial path $p(t,v)=(1-t)p_0+tp_1$ (top), while $\overline\W$ generates displacement interpolation phenomena in which the distribution moves across the map as a function of $t$.}\label{fig:displacement_interpolation}
\end{figure}

Figure~\ref{fig:displacement_interpolation} illustrates displacement interpolation computed between two distributions over a geometric graph of the United States ($|V|=1113, |E|=4058, k=50$).  Whereas the trivial interpolation ``teleports'' mass from one distribution to the other, the time-varying sequence of distributions shows mass moving continuously along the domain.

\subsection{Comparison with $\W_1$}

\begin{figure}\centering
\begin{tabular}{c|ccc}
\includegraphics[width=0.22\textwidth]{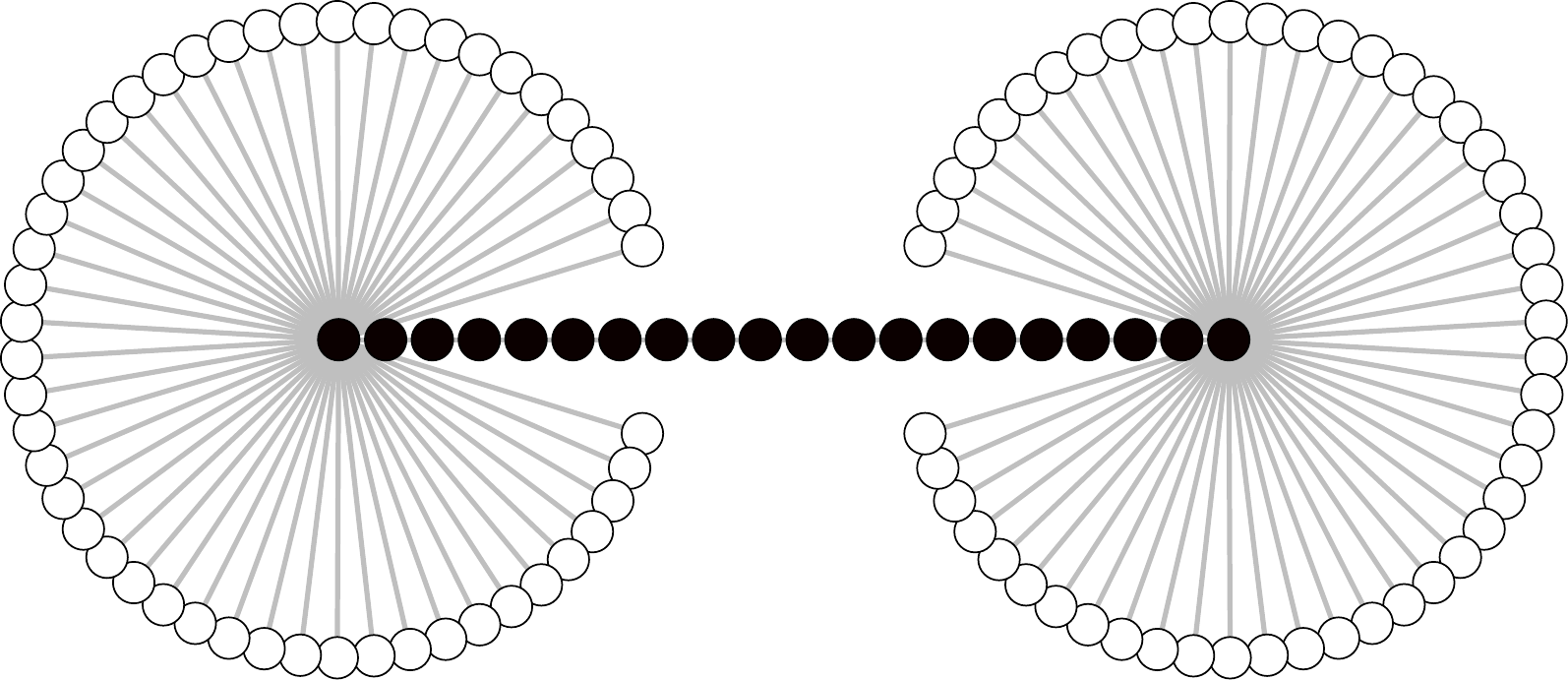}&
\includegraphics[width=0.22\textwidth]{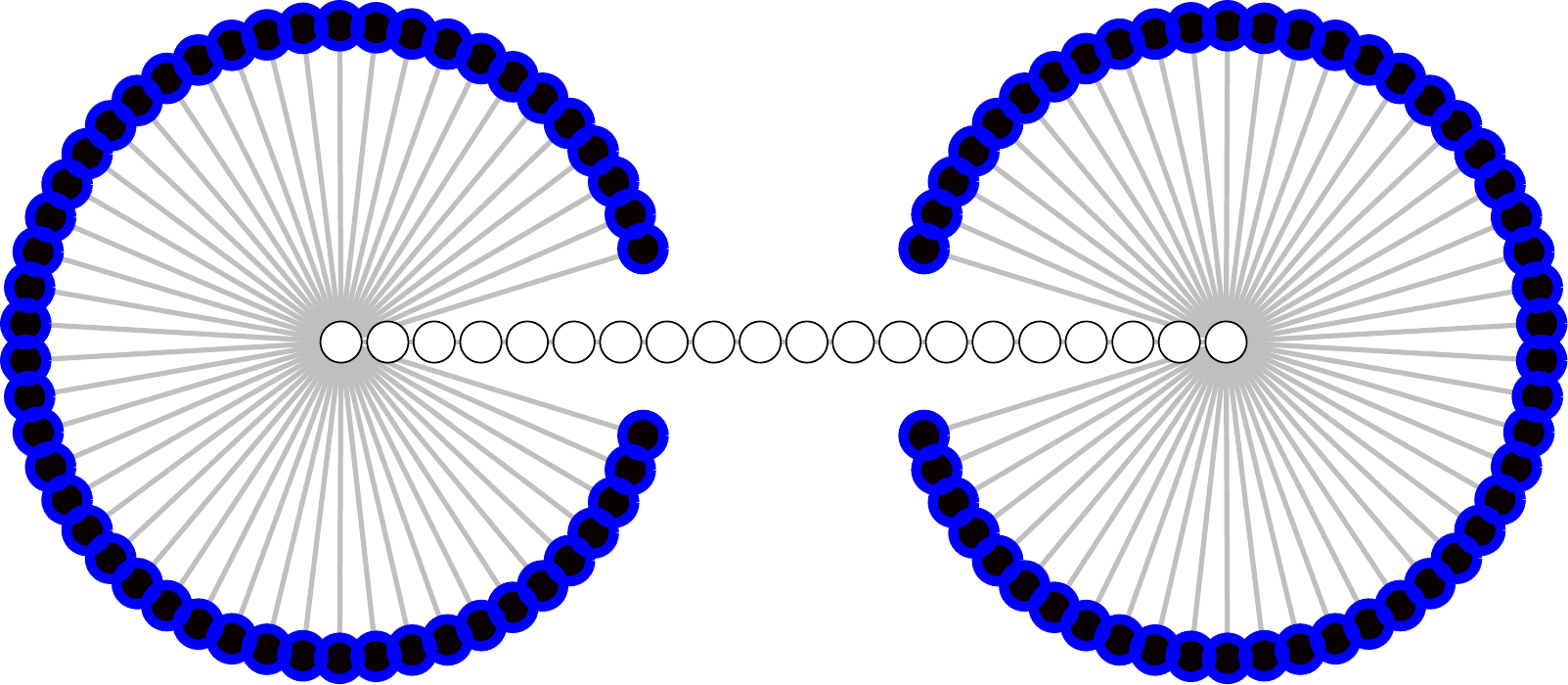}&
\includegraphics[width=0.22\textwidth]{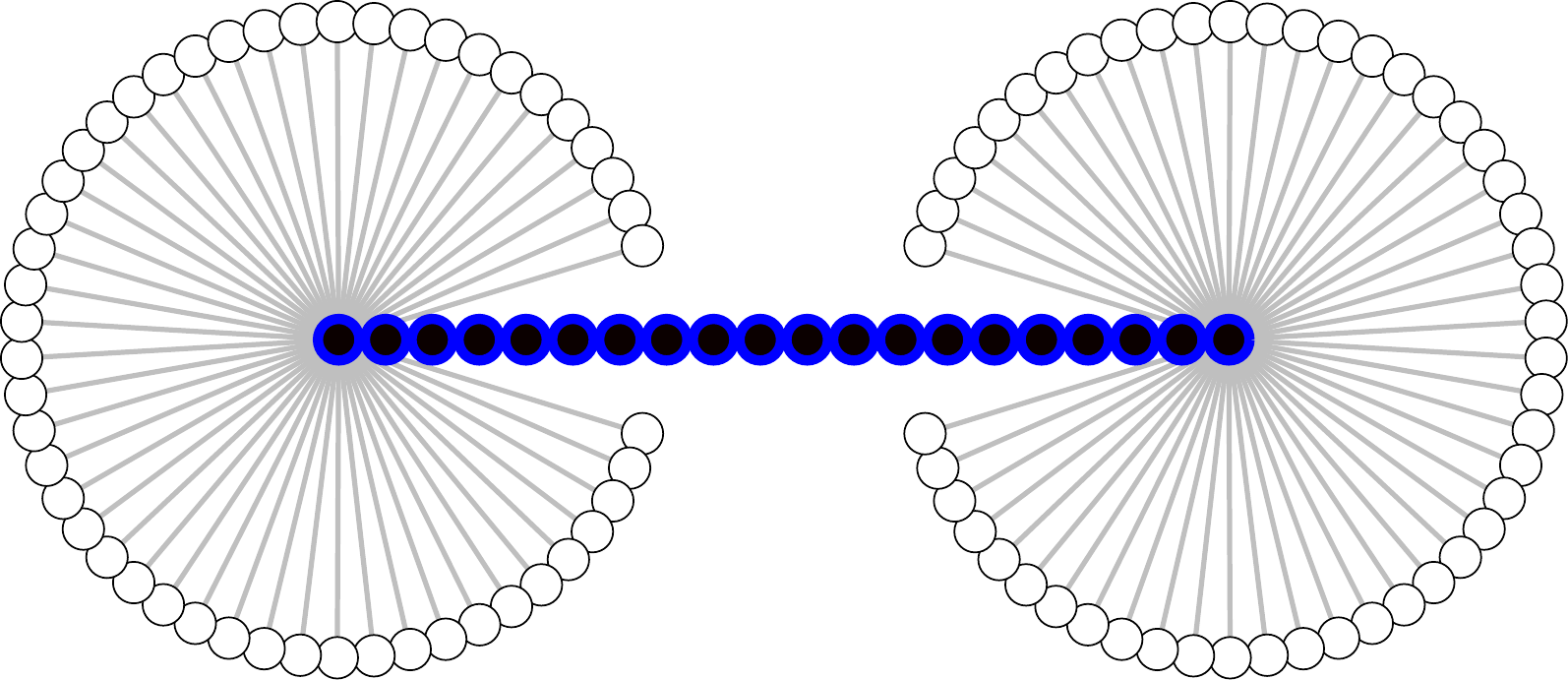}&
\includegraphics[width=0.22\textwidth]{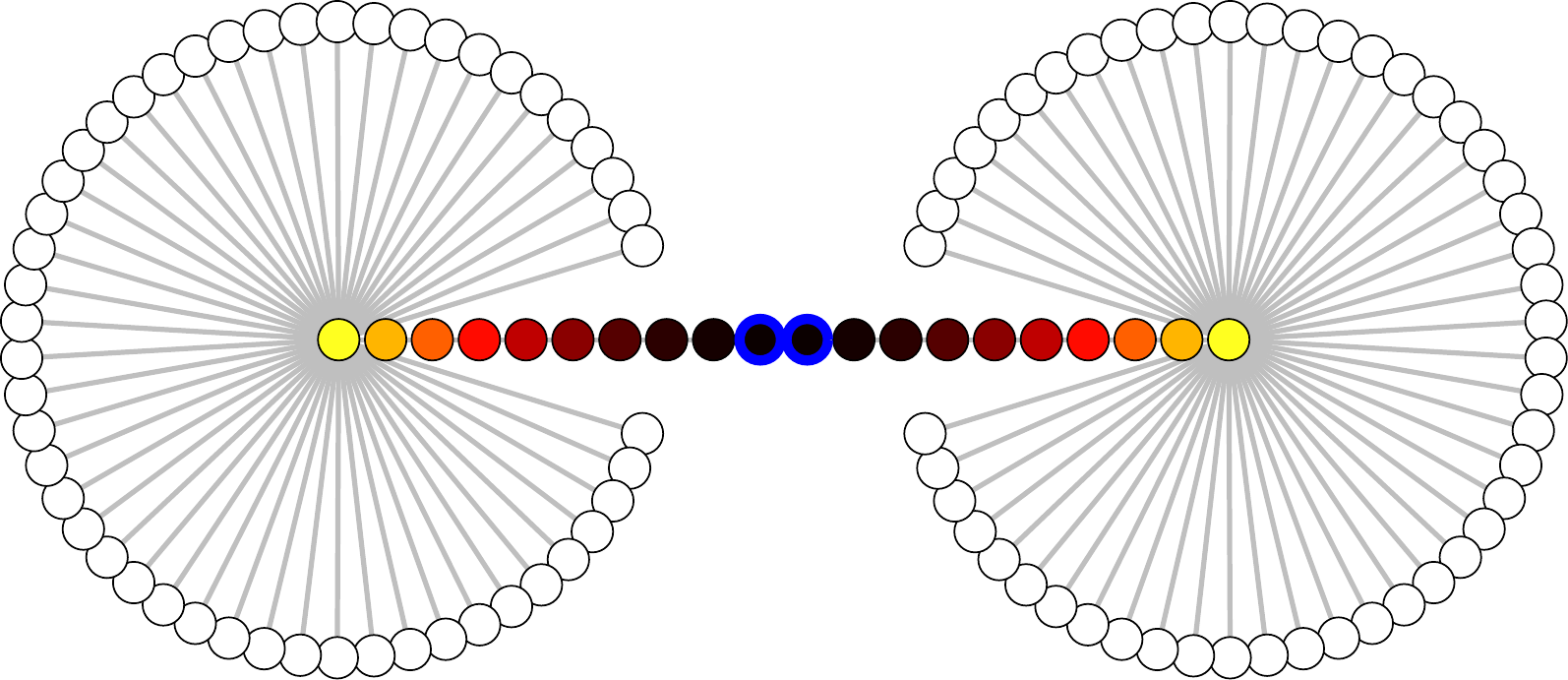}\\
\includegraphics[width=0.22\textwidth]{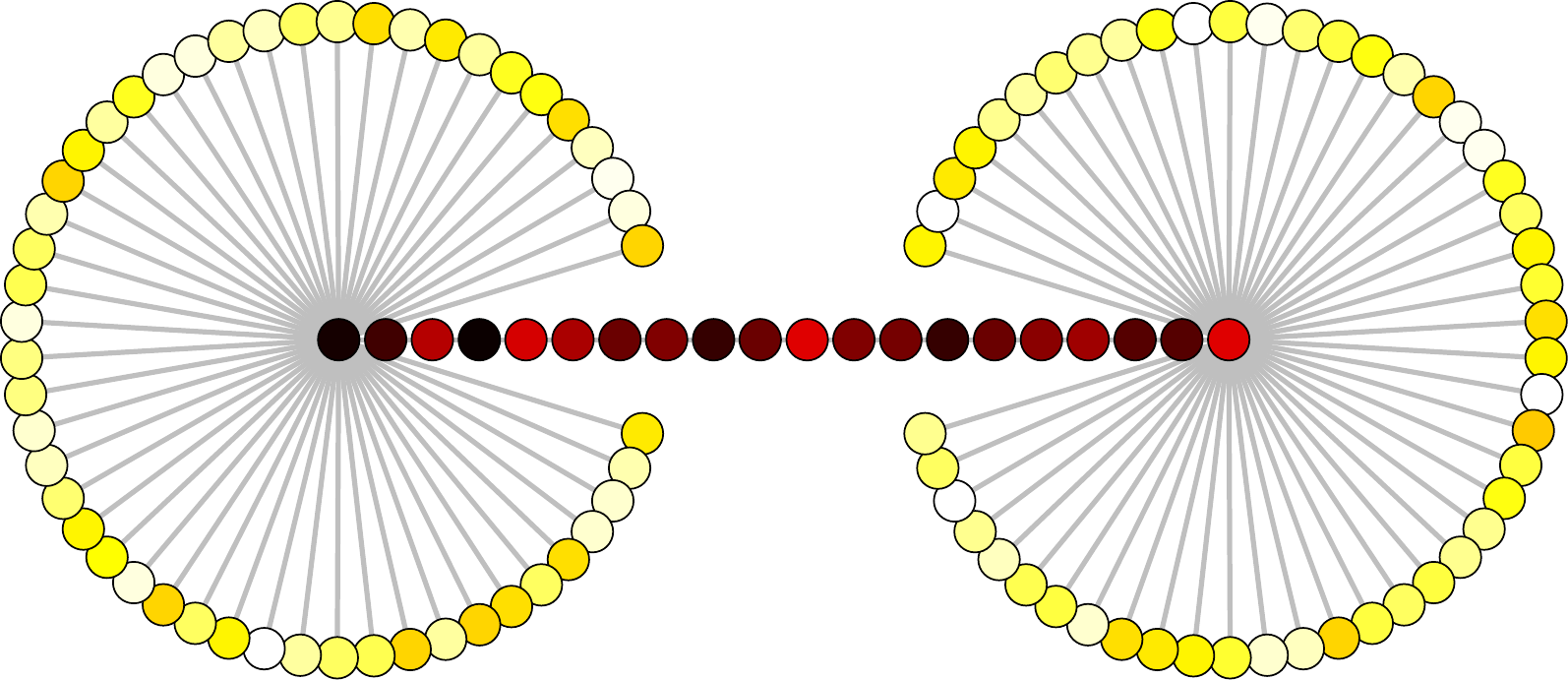}&
\includegraphics[width=0.22\textwidth]{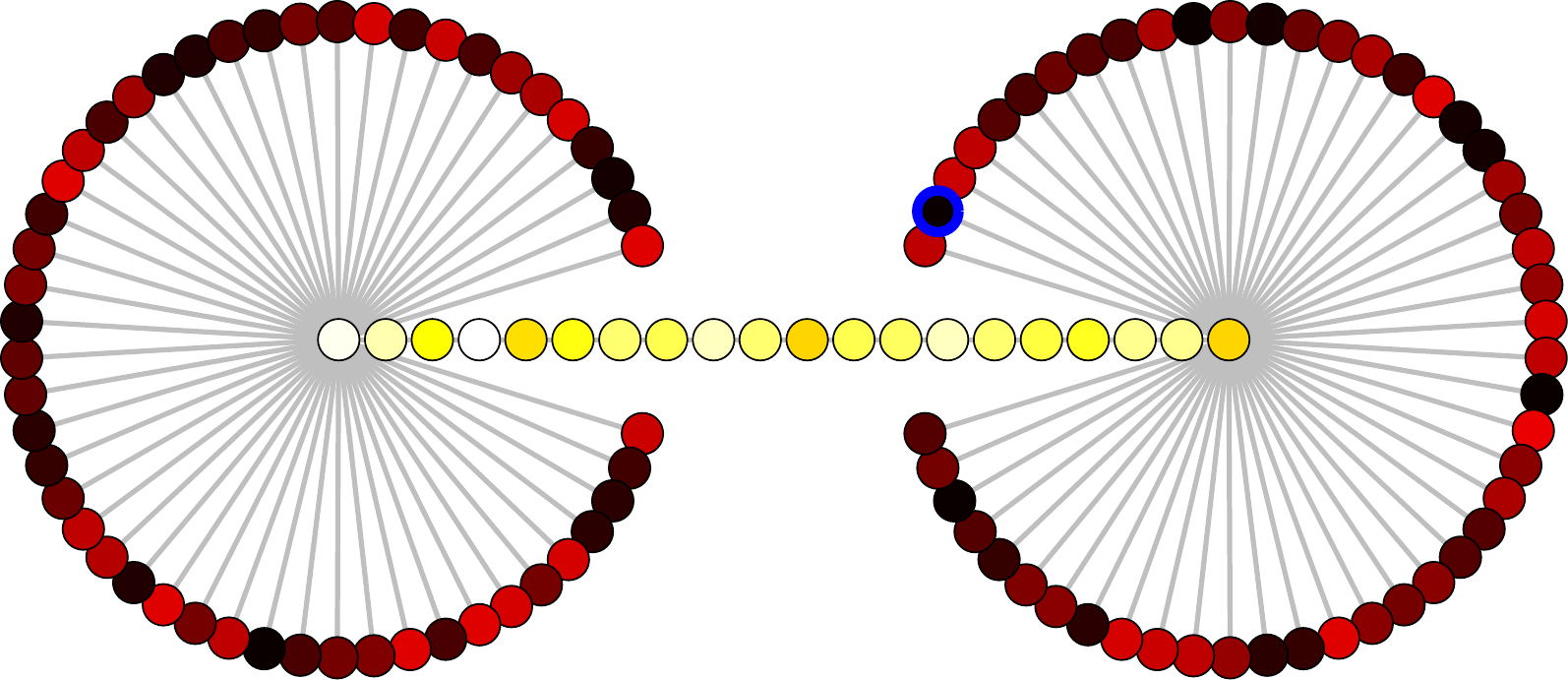}&
\includegraphics[width=0.22\textwidth]{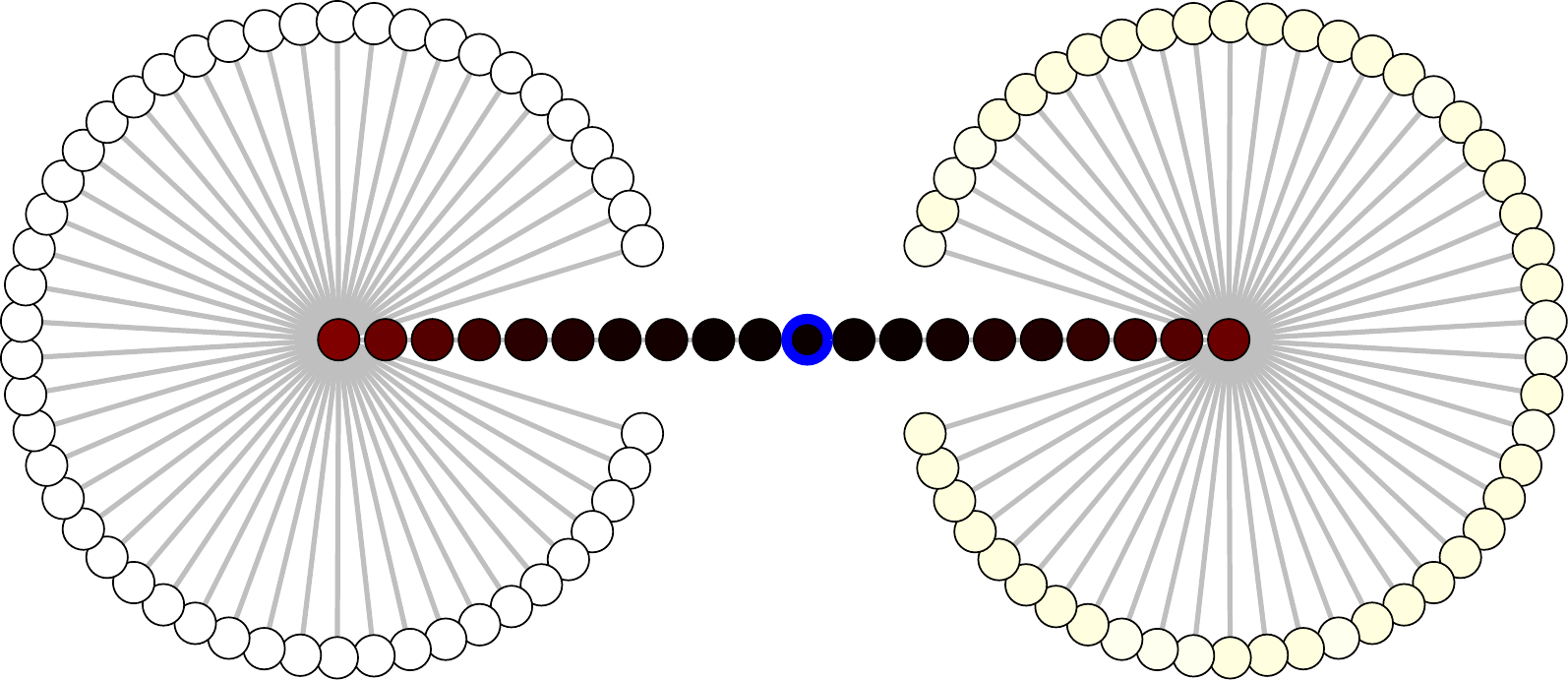}&
\includegraphics[width=0.22\textwidth]{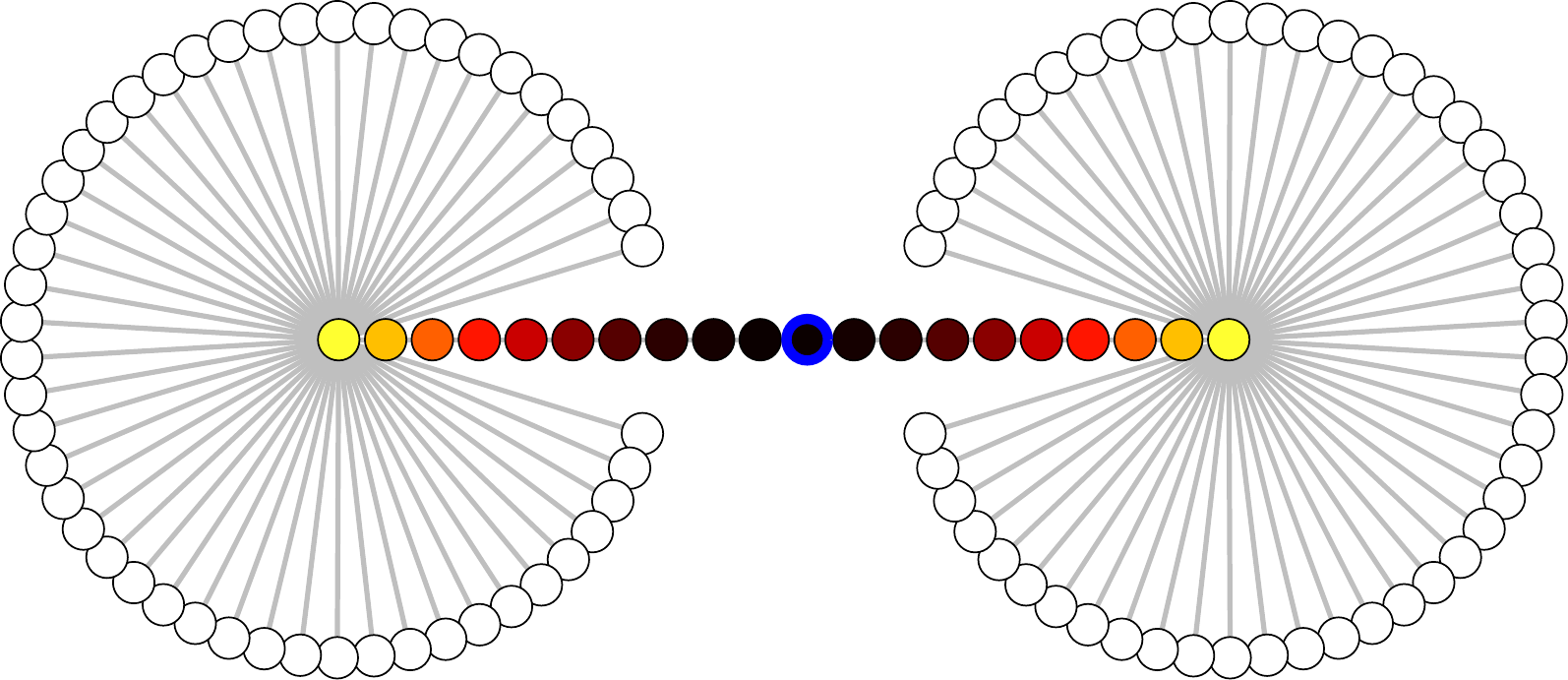}\\
\includegraphics[width=0.22\textwidth]{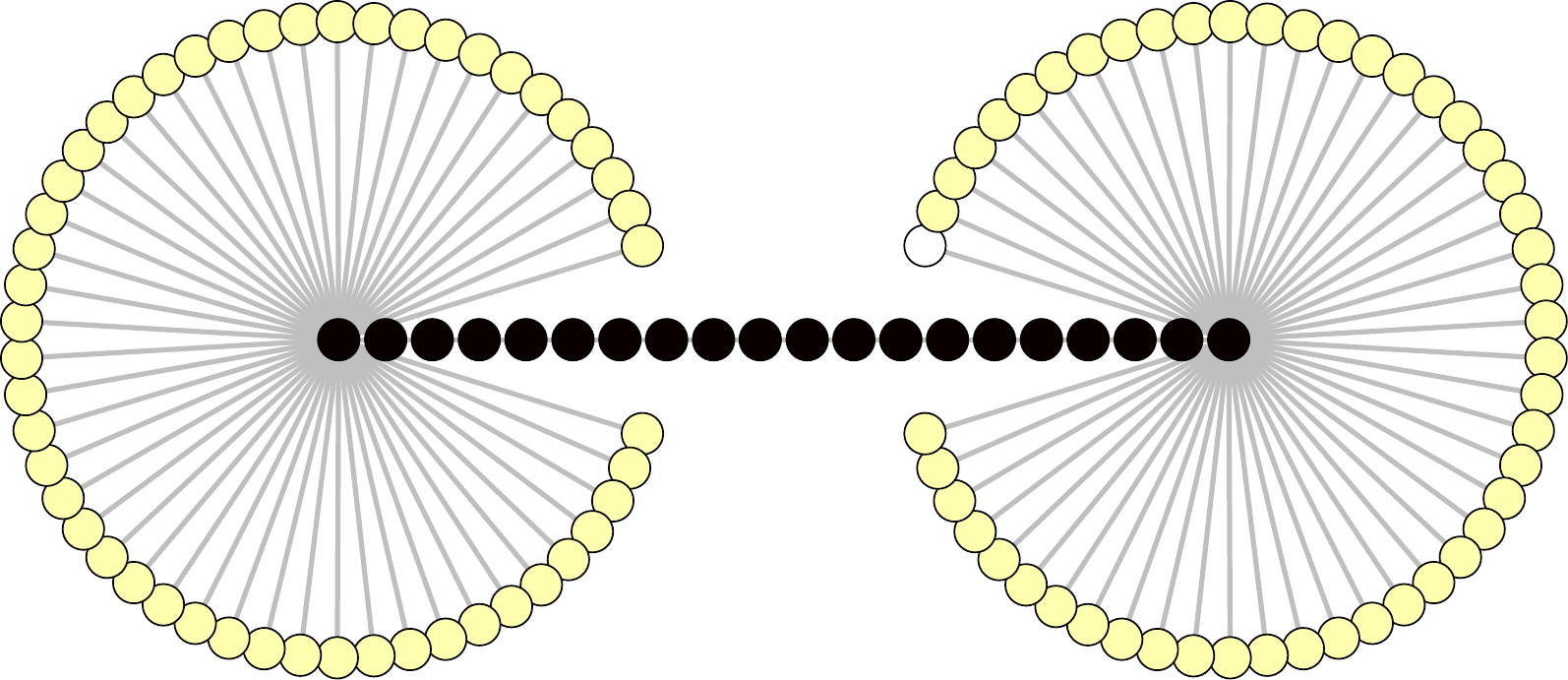}&
\includegraphics[width=0.22\textwidth]{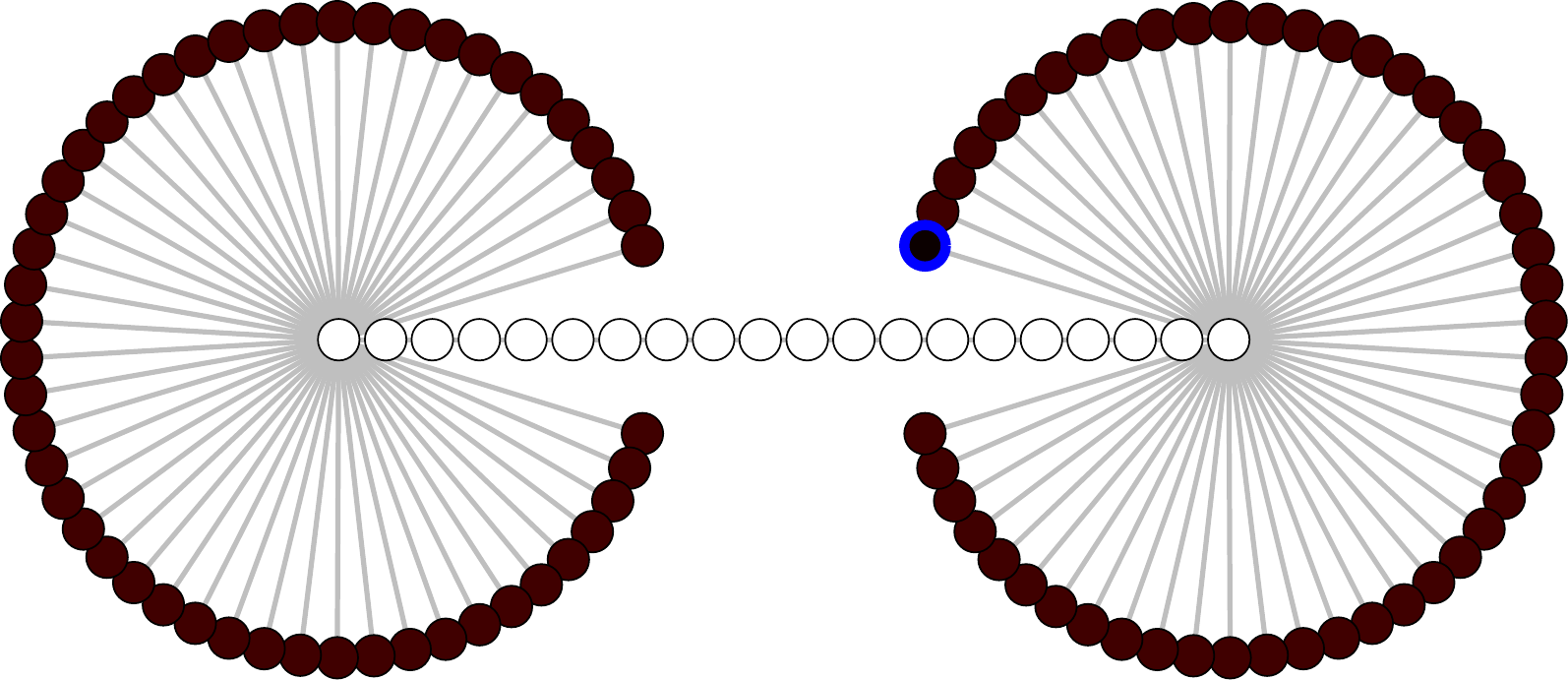}&
\includegraphics[width=0.22\textwidth]{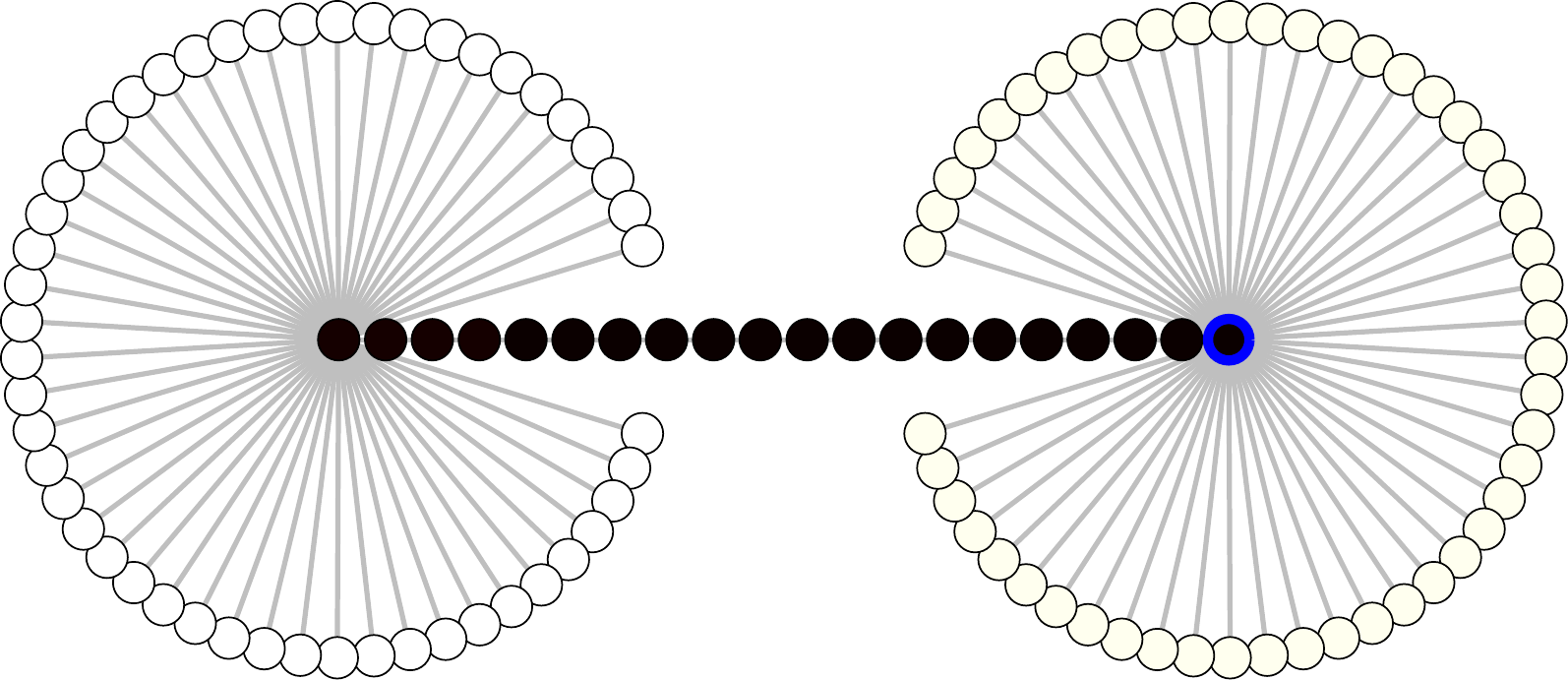}&
\includegraphics[width=0.22\textwidth]{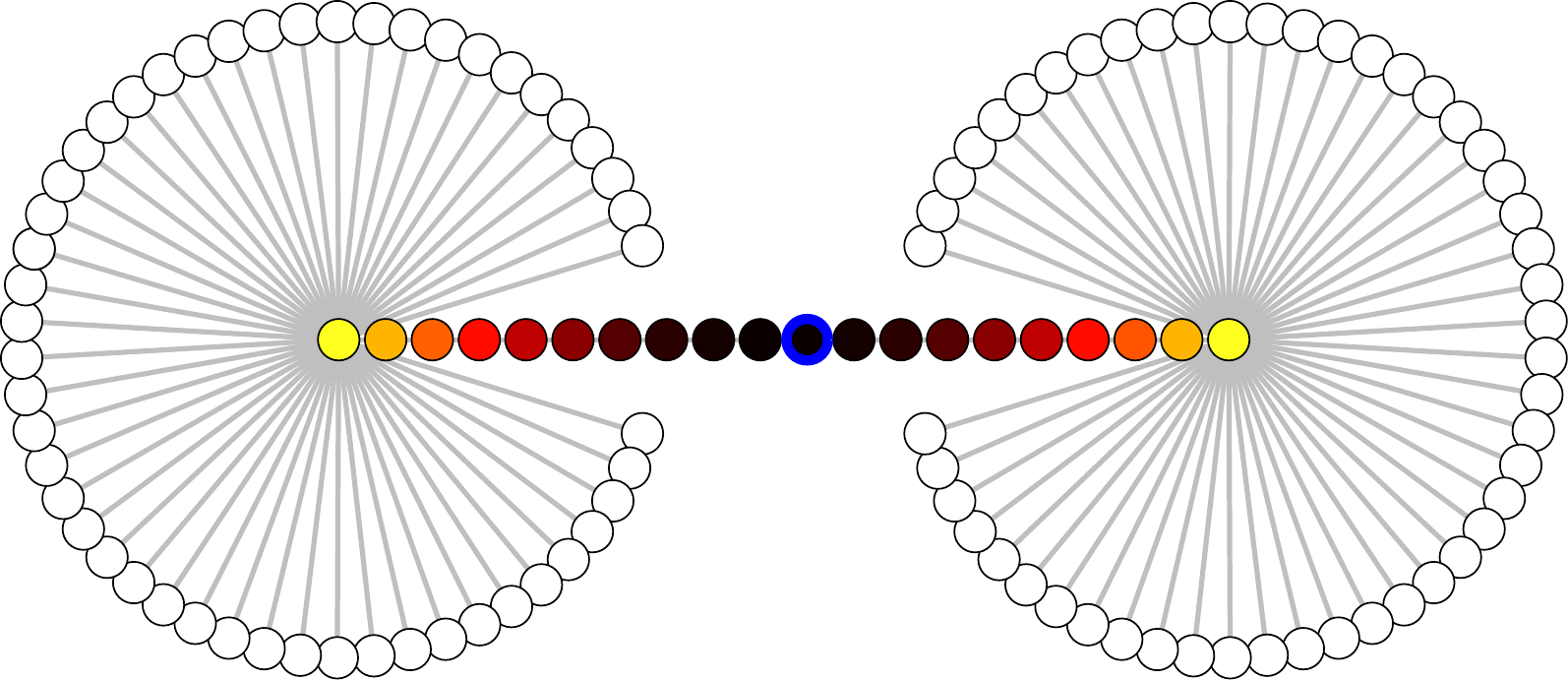}\\
$p(v)$ & $\|p-\delta_v\|_1$ & $\W_1(p,\delta_v)$ & $\overline\W_{20}(p,\delta_v)$
\end{tabular}
\vspace{-.1in}
\caption{Approximating $p$ (left) with the indicator $\delta_v$ of a single vertex $v\in V$.  The distribution $p(v)$ for three experiments is on the left, colored from black ($p(v)=0$) to white ($p(v)=1$); probability is concentrated on the circles uniformly (top), with noise (middle), or with a perturbation on a single vertex (bottom).  Vertices $v$ on the right are colored by $d(p,\delta_v)$ for distance $d$ on the bottom; all $v\in V$ minimizing $d(p,\delta_v)$ are marked in blue.}\label{fig:w1}\vspace{-.2in}
\end{figure}

When the ground distance is the length of the shortest path between vertices, earth mover's distances on a graph are easily computed using multi-commodity network flow; see~\cite{hitchcock-1941} for an early reference.  In particular, for $p,q\in\Prob(G)$, $\W_1(p,q)$ is the minimum of $\sum_{e\in E} |J(e)|$ such that $D^\top J=q-p.$  Hence, we should justify the expense of minimizing~\eqref{eq:discretization}, even if it is cheaper than the na\"ive quadratic distance $\overline\W_{\mathrm{full}}$ in~\eqref{eq:pairwise}.  As evidence beyond that in \S\ref{sec:displacement_interpolation}, we provide an additional experiment here.

Figure~\ref{fig:w1} shows an example constructed to compare $\overline\W$ to $\W_1$ ($|V|=120, |E|=119, k=20$).  Suppose we wish to approximate a distribution $p\in\Prob(G)$ with the indicator $\delta_v$ of a single vertex $v\in V$ by minimizing a probabilistic distance.  Here, $G$ consists of a line of vertices connecting two radial fans, and $p$ is uniform on the fans with zero probability along the line.  $\W_1(p,\delta_v)$ cannot distinguish between the vertices $v$ on the line, and the resulting minimizer $v$ is unstable to perturbations of $p$.  The $L_1$ distance $\|p-\delta_v\|_1$ is nearly minimized by any $v$ with mass on it. $\overline\W$, however, stably chooses $v$ to be the center of the line even in the presence of noise, and the distance function $f(v)\equiv\overline\W(p,\delta_v)$ clearly distinguishes between the vertices of $G$.

Intuition for this test comes from comparing the least squares problem $\min_x (\|x-x_1\|_2^2+\|x-x_2\|_2^2)$ to the geometric median problem $\min_x (\|x-x_1\|_2+\|x-x_2\|_2)$ on $\R^n$.  The former is differentiable with minimizer $\nicefrac{1}{2}(x_1+x_2)$, while the latter is minimized by \emph{any} $x$ on the segment from $x_1$ to $x_2$.

\begin{wraptable}{r}{.47\textwidth}\centering
\vspace{-.15in}
\begin{tabular}{|>{\centering\arraybackslash}m{.3in}|>{\centering\arraybackslash}m{.6in}|>{\centering\arraybackslash}m{.6in}|}\hline
\multirow{2}{*}{$\bm n$} & \multicolumn{2}{c|}{\textbf{\% recovered}}\\\cline{2-3}
& \vspace{.02in}$\bm{\|\cdot\|_1}$ & \vspace{.02in}$\bm{\overline\W}$\\\hhline{|=|=|=|}
\vspace{.02in}10 & \vspace{.02in}64.4\% & \vspace{.02in}64.1\%\\
20 & 73.5\% & 75.7\%\\
30 & 79.4\% & 84.0\%\\
40 & 86.5\% & 89.4\%\\
50 & 89.9\% & 91.9\%\\\hline
\end{tabular}
\caption{Results of shape retrieval task showing the average percentage of shapes in the same category as the query retrieved within the first $n$ results (109 experiments).}\vspace{-.2in}\label{fig:bow}
\end{wraptable}

\subsection{Bag-of-Words Shape Retrieval}

One application of distributions with ground distances is the processing of ``bag-of-words'' models.  Suppose we have a collection of objects with local descriptors, e.g.\ images with per-pixel features.  We can cluster all descriptors from all the domains to yield a sampling of descriptor space; then, each object in the collection can be described as the distribution of how many of its local descriptors have each cluster center as its closest point.  This strategy provides an effective, simple approach to clustering, search, and other problems~\cite{sivic-2005}.

An underlying issue with this model, however, is that the histogram bins are related, especially if the number of clusters is large.  That is, local features may have multiple near-closest cluster centers and could be assigned to different bins equally well.  $L_p$ distances between bag-of-words neglect this structure and hence \emph{degrade} as the number of cluster centers---and correspondingly the sampling of descriptor space---increases beyond a certain point.  Thus, transportation distances may have better performance for such tasks, as proposed in~\cite{jiang-2008,marinai-2011}.

Using extrinsic distance in descriptor space as the transportation ground distance, however, neglects the fact that the set of descriptors may exhibit manifold structure along which distances should be measured intrinsically.  For this reason, rather than measuring Euclidean distance between the cluster centers for the ground distance, we propose using transportation along a graph in which the clusters are linked locally.

Table~\ref{fig:bow} illustrates an experiment for three-dimensional shape retrieval on a database of 109 organic three-dimensional shapes from~\cite{giorgi-2007} classified into 10 categories.  We describe each shape using the bag-of-words method above on wave kernel signature (WKS) descriptors~\cite{aubry-2011} grouped into 1000 clusters; we connect the cluster centers into a connected graph by augmenting the Euclidean minimum spanning tree in WKS space with each center's two nearest neighbors ($|V|=1000, |E|=1668, k=10$).  Compared to using the $L_1$ norm $\|p-q\|_1$ to compare $p,q\in\Prob(G)$, $\overline\W$ consistently achieves better within-class retrieval results using identical descriptors, beyond an initial set of shape matches well-distinguished by any distance.

\section{Discussion and Conclusion}

Our transportation distance $\overline\W$ bridges the gap between continuous and discrete transportation.  Computational techniques on discrete domains like graphs largely are limited to the \emph{one}-Wasserstein distance $\W_1$, which relates to maximum flow problems, but the use of non-squared ground distances creates non-uniqueness and other undesirable properties.  Contrastingly, theoretical understanding of Wasserstein distances largely is limited to the \emph{two}-Wasserstein distance $\W_2$.  Our distance $\overline\W$, however, has properties in common with $\W_2$ without a construction scaling quadratically in $|V|$.

The construction of $\overline\W$ provides a new avenue of research into the theory and practice of transportation distances on graphs.  Most immediately, replacing our use of generic optimization tools with ADMM~\cite{boyd-2011,papadakis-2014} or another specialized technique may help scale $\overline\W$ to graphs like social networks and websites with millions of nodes.  Alternatively, since $D^\top D$ is the Laplacian matrix of $G$, it may be possible to apply spectral techniques to approximate minima of~\eqref{eq:graph_geodesic} e.g.\ by decomposing $J=Da + B$ where $D^\top B=0$~\cite{chung-1994,jiang-2011}.  More generally, our variational construction of $\overline\W$ by introducing a continuous time variable $t$ still led to a distance function over the finite-dimensional space $\Prob(G)$; this pattern may reveal a more general strategy for problems on graphs with more easily-understood continuous analogs.

While our discussion has circulated largely around the construction of $\overline\W$ and some motivating examples, we foresee many practical applications of our machinery and larger optimizations in which it can play a key role.  For instance, the distribution minimizing the sum of squared two-Wasserstein distances yields an effective strategy for computing a \emph{barycenter} summarizing a set of distributions~\cite{bruckstein-2012,cuturi-2013-2,bonneel-2014}; networks of such pairwise distances create effective strategies for semi-supervised learning~\cite{solomon-2014}.  Similar strategies have been applied to construct techniques for clustering~\cite{wagner-2011} and matrix factorization~\cite{sandler-2011}.  Optimal transportation also has had considerable influence on methods in vision~\cite{rubner-2000}, imaging~\cite{bonneel-2014}, graphics~\cite{solomon-2013}, shape processing~\cite{rabin-2010}, road networks~\cite{treleaven-2013}, and other application areas.  With our graph-based treatment we hope to extend these benefits to other domains like documents with geometric embedding structure~\cite{mikolov-2013}.



\small{
\bibliographystyle{abbrv}
\bibliography{graph_distance}
}

\end{document}